\RequirePackage{etoolbox}
\csdef{input@path}{%
 {}
 {}
}%
\csgdef{bibdir}{bib/}

\documentclass[preprint]{imsart}
\usepackage{dsfont} 
\usepackage{amsthm}
\usepackage{amsmath}
\usepackage{amssymb}
\usepackage{natbib}
\usepackage[colorlinks,citecolor=blue,urlcolor=blue,filecolor=blue,backref=page]{hyperref}
\usepackage{graphicx}
\usepackage{bm}
\usepackage{enumitem}
\usepackage{subcaption}

\startlocaldefs
\numberwithin{equation}{section}
\theoremstyle{plain}
\newtheorem{thm}{Theorem}[section]
\newtheorem{corollary}[thm]{Corollay}
\newtheorem{assumption}[thm]{Assumption}
\newtheorem{lemma}[thm]{Lemma}
\newtheorem{remark}[thm]{Remark}
\endlocaldefs

\begin{document}

\begin{frontmatter}
\title{Bayesian Nonparametric Analysis of Multivariate Time Series: A Matrix Gamma Process Approach}
\runtitle{Bayesian Nonparametric Analysis of Multivariate Time Series}

\begin{aug}
\author{\fnms{Alexander} \snm{Meier}\thanksref{addr1}\ead[label=e1]{alexander.meier@ovgu.de}},
\author{\fnms{Claudia} \snm{Kirch}\thanksref{addr1}\ead[label=e2]{claudia.kirch@ovgu.de}}
\and
\author{\fnms{Renate} \snm{Meyer}\thanksref{addr2} \ead[label=e3]{renate.meyer@auckland.ac.nz}}

\runauthor{A. Meier et al.}

\address[addr1]{Institute for Mathematical Stochastics, Department of Mathematics,
     Otto-von-Guericke University Magdeburg
    \printead{e1} 
    \printead*{e2}
}

\address[addr2]{Department of Statistics,
    The University of Auckland
    \printead{e3}
}

\end{aug}

\begin{abstract}
While there is an increasing amount of literature about Bayesian time series analysis, 
only few Bayesian nonparametric approaches to multivariate time series exist. 
Most methods rely on Whittle's Likelihood, involving the second order structure of a stationary time series by means
of its spectral density matrix. 
This is often modeled in terms of the Cholesky decomposition to ensure positive definiteness. 
However, asymptotic properties such as posterior consistency or posterior contraction rates are not known.
A different idea is to model the spectral density matrix by means of random measures. 
This is in line with existing approaches for the univariate case, where the normalized spectral density is
modeled similar to a probability density, e.g. with a Dirichlet process mixture of Beta densities.
In this work we present a related approach for multivariate time series, with matrix-valued mixture weights induced by a Hermitian positive definite Gamma process. 
The proposed procedure is shown to perform well for both simulated and real data. 
Posterior consistency and contraction rates are also established.
\end{abstract}

\begin{keyword}[class=MSC]
\kwd[Primary ]{62G20} 
\kwd{62M10} 
\kwd{60G57} 
\kwd[; secondary ]{60G15} 
\end{keyword}

\begin{keyword}
\kwd{Bayesian nonparametrics}
\kwd{Completely random measures}
\kwd{Spectral density}
\kwd{Stationary multivariate time series}
\end{keyword}

\end{frontmatter}


\newcommand{\conj}[1]{\overline{#1}}
\newcommand{\adj}[1]{{#1}^*} 
\newcommand{\param}[1]{#1_{\operatorname{param}}}
\newcommand{\paramm}[2]{#1_{\operatorname{param};#2}}
\newcommand{\outerProduct}[1]{#1\adj{#1}}
\newcommand{\distr}{\sim}
\newcommand{\asEq}{\overset{\text{a.s.}}{=}}
\newcommand{\dEq}{\overset{d}{=}}
\newcommand{\iid}{\overset{\text{iid}}{\distr}}
\newcommand{\ind}{\overset{\text{ind.}}{\distr}}
\newcommand{\vech}{\operatorname{vec}}
\newcommand{\vecop}[1]{{{#1}}_{\operatorname{vec}}}
\newcommand{\veccop}[1]{{{#1}}_{\underline{\operatorname{vec}}}}
\newcommand{\diag}{\operatorname{diag}}
\newcommand{\Id}{\operatorname{Id}}
\newcommand{\Dir}{\operatorname{Dir}}
\newcommand{\tr}{\operatorname{tr}}
\newcommand{\block}{\mathcal B}
\newcommand{\etr}{\operatorname{etr}}
\newcommand{\lvy}{L\'{e}vy }
\newcommand{\frch}{Fr\'{e}chet }
\newcommand{\Exp}{\operatorname{Exp}}
\newcommand{\indi}{\mathds{1}}
\newcommand{\lleq}{\lesssim}
\newcommand{\ggeq}{\gtrsim}
\newcommand{\vardbtilde}[1]{\tilde{\raisebox{0pt}[0.85\height]{$\tilde{#1}$}}}
\newcommand{\matr}[1]{\bm{#1}} 
\newcommand{\matrt}[1]{\tilde{\bm{#1}}} 
\newcommand{\vect}[1]{\underline{#1}} 
\newcommand{\vectt}[1]{\underline{\tilde{#1}}} 
\newcommand{\ups}[1]{^{(#1)}}
\newcommand{\upss}[1]{^{[#1]}}
\newcommand{\cblue}{\color{blue}}
\newcommand{\cred}{\color{red}}

\newcommand{\modd}{\operatorname{mod}}
\newcommand{\MF}{\mathcal{MF}}

\newcommand{\E}{\mathrm E}
\newcommand{\Var}{\mathrm{Var}}
\newcommand{\Cov}{\mathrm{Cov}}
\newcommand{\DFT}{\mathcal F_n}
\newcommand{\iDFT}{\DFT^{-1}}

\section{Introduction}

With the surge of research in Bayesian nonparametrics over the last two decades, several nonparametric Bayesian approaches to analysing univariate stationary time series have been developed, such as \citet{CarterKohn97}, \citet{Gango98},
\citet{Liseo01}, \citet{choudhuri}, \citet{Hermansen08}, \cite{Chopin13}, \cite{Cadonna2017} and \cite{Edwards2018}. These are generally based on Whittle's likelihood \citep{Whittle57}, 
essentially the product of the asymptotically independent discrete Fourier transformed observations with variances equal to the spectral density at the corresponding Fourier frequencies, providing an approximation to the true likelihood. The Whittle likelihood
avoids matrix inversion (for Gaussian time series) and significantly speeds up computation.  Even for non-Gaussian time series, it provides a reasonable approximation and asymptotically correct 
inference in many situations \citep{Shao07}. The nonparametric prior on the spectral density used in \cite{choudhuri} is based on the Bernstein polynomial approximation
first employed in \cite{petrone1999random} for nonparametric density estimation.  Whereas no asymptotic results were provided for  other nonparametric priors
such as those based on state-space methodology by \citet{CarterKohn97} or fitting local polynomials to the log spectral density in \citet{Gango98}, \cite{choudhuri} proved consistency 
of the posterior distribution based on the Bernstein polynomial mixtures and the Whittle likelihood for Gaussian time series. 

The Whittle likelihood for multivariate time series is a function of the spectral density matrix,  a Hermitian positive definite (Hpd) matrix-valued function. 
Several frequentist and  Bayesian extensions of nonparametric spectral density estimation to multivariate time series have been suggested in the recent literature (\cite{dai2004multivariate}, \cite{rosen2007automatic} and \cite{li2017adaptive}).
To ensure positive definiteness,
these approaches are based on smoothing the components of the Cholesky decomposition of the periodogram matrix.
 \cite{rosen2007automatic} proposed  a
Bayesian nonparametric approach based on fitting smoothing splines to each real and imaginary component of the Cholesky decomposition of the periodogram matrix. Similarly, \cite{zhang2016adaptive}
proposed an extension to nonstationary multivariate time series by an adaptive estimation of time-varying spectra and  cross-spectra  based on the Cholesky decomposition of the inverse of the spectral density matrix. \cite{Zhang2018} described an extension to quantile-based spectra based on the Cholesky decomposition of the copula spectral density kernels. 
However, using the Cholesky decomposition of the spectral density matrix will make it difficult to elicit priors. 
While simulations studies and applications to real data generally demonstrate good performance of these extensions, no asymptotic results for any
Bayesian nonparametric approaches to multivariate time series are known.

Therefore, with the aim of proving posterior consistency of a nonparametric Bayesian approach to multivariate time series, we extend the Bernstein-Dirichlet process prior of \cite{choudhuri} from the space of positive real-valued functions to the space of Hpd matrix-valued functions. We employ the equivalence of the  Dirichlet process and the normalized Gamma process used to define the mixture weights of the Bernstein polynomial mixture and generalize the completely random measure (CRM) on ${\cal X}=[0,\pi]$ with independent Gamma increments to a completely random Hpd matrix-valued measure with independent matrix-Gamma increments. To this end, we first define an infinitely divisible Hpd Gamma distribution in terms of the L\'{e}vy-Khinchine representation of its Laplace transform, following
\cite{perez}. Then, in analogy to the Gamma process, an Hpd Gamma process is defined using the Kingman construction \citep{kingman1992poisson} of a CRM with Hpd Gamma increments.
A software implementation of the methodology is included in the \verb|R| package \verb|beyondWhittle|, 
which is available on the Comprehensive \verb|R| Archive Network (CRAN), see~\cite{beyondWhittlePackage}.

The paper is organized as follows: Section 2 first introduces an infinitely divisible Hpd Gamma distribution defined via the L\'{e}vy-Khinchine representation of its Laplace transform. It then
details the construction of the Hpd Gamma process using the Kingman construction based on a Poisson process with a suitably defined intensity measure  corresponding to the L\'{e}vy measure
of a Hpd Gamma distribution. It is proven that this defines a completely random Hpd measure with independent Hpd Gamma increments. Furthermore, we derive an almost surely convergent series representation of the Hpd Gamma process. Section 3 provides a hierarchical prior for the spectral density matrix of a multivariate stationary time series using the
mixture of Bernstein polynomial prior with matrix-valued weights induced by a Hpd Gamma process. The posterior distribution is obtained by combining the prior with the Whittle likelihood which is sampled by an efficient Inverse L\'{e}vy measure algorithm as in \cite{wolpert1998simulation}. In Section 4, we essentially follow the steps in \cite{choudhuri} and first show mutual contiguity of Whittle likelihood and full Gaussian likelihood under certain conditions on the Gaussian stationary time series. We then prove posterior consistency in the $\mathbb L^1$ norm and strengthen the result to uniform consistency under the additional assumption of uniform boundedness of the Lipschitz constants of the spectral density. Section 4 also establishes posterior contraction rates in the Hellinger distance. The performance of the proposed nonparametric Bayesian approach is illustrated in Section 5 in a simulation study where it is compared to a parametric vector-autoregressive model. It is then applied to the analysis of the bivariate monthly time series of the Southern Oscillation Index and fish recruitment from 1950--1987, previously analyzed by \cite{rosen2007automatic}
and \cite{shumway2010time}. Section 6 concludes with an outlook to the analysis of non-Gaussian time series,  frequentist coverage probabilities of Bayesian credible sets and alternative applications of the Hpd Gamma process.

\section{Hpd Gamma Process}\label{sec:process}
In this section, we construct an Hermitian positive definite (Hpd) Gamma process
that will be used in the Bayesian procedure described in Section~\ref{sec:bayesProcedure}.
We start our considerations by revisiting the Hpd Gamma distribution from~\cite{perez}.
Let us first introduce some notation.
For a complex~$d\times d$ matrix~$\matr A$, let~$|\matr A|:=|\det \matr A|$. 
Denote the trace of~$\matr A$ by~$\tr \matr A$, let~$\etr\matr A:=\exp(\tr\matr A)$ and denote the Hermitian conjugate by~$\matr A^*$.
The real- and imaginary part of~$\matr A$ will be denoted by~$\Re\matr A$ and~$\Im\matr A$.
The identity matrix is denoted by~$\matr I_d$.
Denote by~$\mathcal S_d$ the space of Hermitian matrices in~$\mathbb C^{d\times d}$ and by~$\mathcal S_d^+$ the open cone of Hpd matrices.
Furthermore, let~$\bar{\mathcal S}_d^+$ denote the closed cone of Hermitian positive semidefinite (Hpsd) matrices.
For~$\matr Z \in \bar{\mathcal S}_d^+$, denote the (unique) Hpsd square root by~$\matr Z^{1/2}\in \bar{\mathcal S}_d^+$.
We will consider the Frobenius norm~$\|\matr A\|:=\sqrt{\tr(\matr A\matr A^*)}$
and the trace norm~$\|\matr A\|_T := \tr( (\matr A \matr A^*)^{1/2} )$.
For~$\matr Z \in \bar{\mathcal S}_d^+$, the trace norm simplifies to~$\|\matr Z\|_T = \tr\matr Z$.
Denote by~$\mathbb S_d^+ = \{ \matr U \in \mathcal S_d^+ \colon \tr \matr U = 1 \}$ the open unit sphere in~$\mathcal S_d^+$ with respect to the trace norm.
The closure of~$\mathbb S_d^+$ is denoted by~$\bar{\mathbb S}_d^+$.
We will often work with the representation~$\bar{\mathcal S}_d^+ \cong \bar{\mathbb S}_d^+ \times (0,\infty)$, decomposing~$\bar{\mathcal S}_d^+\ni\matr Z=r\matr U$
into a \emph{spherical part}~$\matr U \in \bar{\mathbb S}_d^+$ and a \emph{radial part}~$r=\tr \matr Z\in (0,\infty)$.
For~$\matr A \in \mathcal S_d$, denote the (ordered) eigenvalues by~$\lambda_1(\matr A) \leq \ldots \leq \lambda_d(\matr A)$.

\subsubsection*{Infinitely divisible Hpd Gamma distribution}
Let~$\alpha$ be a finite measure on~$\bar{\mathbb S}_d^+$ and let~$\beta \colon \bar{\mathbb S}_d^+ \to (0,\infty)$ be a measurable function,
such that the integral~$\int_{\bar{\mathbb S}_d^+}\log(1+1/\beta(\matr U))\alpha(d\matr U)$ is finite.
The Hpd Gamma distribution~$\operatorname{Ga}_{d\times d}(\alpha,\beta)$ with parameters~$\alpha,\beta$
is defined in terms of the L\'{e}vy-Khinchine representation of the Laplace transform
as follows:~$\matr X \sim \operatorname{Ga}_{d\times d}(\alpha,\beta)$ if and only if
\[ 
  \E\etr(-\matr\Theta\matr X)
  = \exp\left( -\int_{\bar{\mathbb S}_d^+}\int_0^\infty (1-\etr(-r\matr\Theta\matr U))\frac{\exp(-\beta(\matr U)r)}{r}dr\alpha(d\matr U) \right)
\]
holds for all~$\matr\Theta\in\bar{\mathcal S}_d^+$.
Under the above assumptions on~$\alpha$ and~$\beta$, the~$\operatorname{Ga}_{d\times d}(\alpha,\beta)$ distribution is well-defined and it holds~$P(\matr X\in\bar{\mathcal S}_d^+)=1$.
Observe that for~$d=1$, the Hpd Gamma distribution coincides with the Gamma distribution.
The parameters~$\alpha,\beta$ are generalizations of the scale and rate parameter.
Since the $\operatorname{Ga}_{d\times d}(\alpha,\beta)$ distribution is defined in terms of the L\'{e}vy-Khinchine representation of its Laplace transform, it is necessarily infinitely divisible (see Theorem~1 and Remark~2 in~\cite{perez2007representation}).
This constitutes the key property for the upcoming Hpd Gamma process construction.
We emphasize that e.g.~the complex Wishart distribution (being another well-known generalization of the Gamma distribution to the Hpd matrix case) lacks this property~\citep{levy1948arithmetical}.
Observe that the \lvy measure of the~$\operatorname{Ga}_{d\times d}(\alpha,\beta)$ distribution on~$\bar{\mathcal S}_d^+ \cong \bar{\mathbb S}_d^+ \times (0,\infty)$ is given by
\begin{equation}\label{eq:levyMeasureGammaDistr}
  \nu(d\matr U,dr) = \frac{\exp(-\beta(\matr U)r)}{r}dr\alpha(d\matr U).
\end{equation}
\begin{remark}\label{rem:aGamma}
A special case of the~$\operatorname{Ga}_{d\times d}(\alpha,\beta)$ distribution
is the so-called~$A\Gamma$ distribution, that has been considered in~\cite{perez}
and generalized to the Hpd setting in~\cite{meier}, Section~2.4.
To elaborate, the $A\Gamma(\eta,\omega,\matr\Sigma)$ distribution is defined
with the parameters~$\eta>d-1$, $\omega>0$ and~$\matr\Sigma\in\mathcal S_d^+$
as the~$\operatorname{Ga}_{d\times d}(\alpha_{\eta,\matr\Sigma},\beta_{\matr\Sigma})$ distribution, 
with
\[
  \alpha_{\eta,\matr\Sigma}(d\matr U)=|\matr\Sigma|^{-\eta}\tr(\matr\Sigma^{-1}\matr U)^{-d\eta}\Gamma(d\eta)\tilde\Gamma_d(\eta)^{-1}|\matr U|^{\eta-d}d\matr U,
\]
where~$\Gamma$ denotes the Gamma function and~$\tilde\Gamma_d$
the complex multivariate Gamma function (see~\cite{mathai2005some}),
and~$\beta_{\matr\Sigma}(\matr U)=\tr(\matr\Sigma^{-1}\matr U)$.
It has the advantage that for~$\matr X\sim A\Gamma(\eta,\omega,\matr\Sigma)$,
the formulas for mean and covariance structure are explicitly known:
\[
  \E\matr X = \frac{\omega}{d}\matr\Sigma,\quad
  \Cov\matr X = \frac{\omega}{d(\eta d+1)}(\eta\matr I_{d^2}+\matr H)(\matr\Sigma\otimes\matr\Sigma),
\]
where~$\matr H=\sum_{i,j=1}^d\matr H_{i,j}\otimes H_{j,i}$ and~$\matr H_{i,j}$ being the matrix
having a one at~$(i,j)$ and zeros elsewhere, see~\cite{meier}, Lemma~2.8.
Thus the~$A\Gamma$-distribution is particularly well suited for Bayesian
prior modeling if the prior knowledge is given in terms of mean and covariance structure.
\end{remark}

\subsubsection*{Process construction}
We can now utilize the infinitely divisible Hpd Gamma distribution to define an Hpd Gamma process, 
i.e.~a stochastic process with independent~$\operatorname{Ga}_{d\times d}(\alpha,\beta)$ distributed increments.
The process construction is based on Poisson processes and generalizes the famous Kingman construction of the Gamma process 
(see Section~8.2 in~\cite{kingman1992poisson}) to the Hpd matrix case.
Let us briefly recall the notion of Poisson processes.
For a Borel space~$\mathcal Y$, a \emph{Poisson process}~$\Pi$ on~$\mathcal Y$ is a countable subset such that for all~$m>0$ 
and all disjoint subsets~$A_1,\ldots,A_m$, it holds that~$\# \{\Pi\cap A_1\}, \ldots, \# \{\Pi\cap A_m\}$ are independent random variables, 
distributed as~$\operatorname{Poi}(\nu(A_j))$ for~$j=1,\ldots,m$.
Here,~$\nu$ is a measure on~$\mathcal Y$, which is called the \emph{mean measure} of~$\Pi$
and we write~$\Pi \sim \operatorname{PP}(\nu)$.
For a rigorous treatment of Poisson processes, the reader is referred to~\cite{kingman1992poisson}.

\par
Let~$\mathcal X$ be a Polish space,
equipped with a locally compact,~$\sigma$-finite and nontrivial Borel measure.
To define a~$\operatorname{Ga}_{d\times d}(\alpha,\beta)$ process on~$\mathcal X$, we allow the distributional parameters~$\alpha,\beta$ to vary among~$\mathcal X$.
To elaborate, denote by~$\mathbb B(\bar{\mathbb S}_d^+)$ the Borel sets in~$\bar{\mathbb S}_d^+$. 
Let~$\alpha \colon \mathcal X \times \mathbb B(\bar{\mathbb S}_d^+) \to [0,\infty)$ such that~$\{ \alpha(x,\cdot) \}_{x\in\mathcal X}$ is
a family of finite measures on~$\bar{\mathbb S}_d^+$ and for all~$B \in \mathbb B(\bar{\mathbb S}_d^+)$ the mapping~$\mathcal X \ni x \mapsto \alpha(x,B)$ is measurable.
Furthermore, let~$\beta \colon \mathcal X \times \bar{\mathbb S}_d^+ \to (0,\infty)$ be measurable.
Define the measure~$\nu$ on~$\mathcal X \times \bar{\mathcal S}_d^+$ as
\begin{equation}\label{eq:meanMeasureProcess}
  \nu(dx,d\matr U,dr) = \frac{\exp(-\beta(x,\matr U)r)}{r}dr\alpha(x,\matr U)dx.
\end{equation}
For~$x \in \mathcal X$,
the measure~$\nu(d\matr U,dr|x):=\frac{1}{r}\exp(-\beta(x,\matr U)r)dr\alpha(x,d\matr U)$ on~$\bar{\mathcal S}_d^+$ corresponds to the \lvy
measure of the~$\operatorname{Ga}_{d\times d}(\alpha(x,\cdot),\beta(x,\cdot))$ distribution,
see~\eqref{eq:levyMeasureGammaDistr}.
In what follows, we will make the following assumption on~$\nu$ from~\eqref{eq:meanMeasureProcess}: 
\begin{equation}\label{eq:nuAssumption}
  \int_{\mathcal X}\int_{\bar{\mathbb S}_d^+}\int_0^\infty\min(1,r)\nu(dx,d\matr U,dr)<\infty.
\end{equation}
This property ensures that~$\nu$ is a feasible Poisson process mean measure (see Section~2.5 in~\cite{kingman1992poisson}).
Let~$\Pi\sim\operatorname{PP}(\nu)$ and define the process~$\matr\Phi \colon \mathbb B(\mathcal X) \to \bar{\mathcal S}_d^+$ as
\begin{equation}\label{eq:phiCRMdef}
  \matr\Phi(A):=\sum_{(x,\matr U,r)\in\Pi} \indi_A(x)r\matr U, \quad A \subset \mathcal X \text{ measurable}.
\end{equation}
The following result shows that~$\matr\Phi$ is well-defined and an independent increment process, with~$\operatorname{Ga}_{d\times d}$ distributed increments.
\begin{thm}\label{th:distOfPhi}
  Let~$\matr\Phi$ be defined as in~\eqref{eq:phiCRMdef}, with~$\Pi\sim\operatorname{PP}(\nu)$
  and~$\nu$ from~\eqref{eq:meanMeasureProcess} fulfilling assumption~\eqref{eq:nuAssumption}.
  Then it holds:
  \begin{enumerate}[label=(\alph*)]
  \item For all measurable~$A \subset \mathcal X$ it holds~$P(\matr\Phi(A)\in\bar{\mathcal S}_d^+)=1$.
  The distribution of~$\matr\Phi(A)$ is given in the L\'{e}vy-Khinchine representation
  \[
    \E \etr(-\matr\Theta\matr\Phi(A))=\exp\left( -\int_{\bar{\mathbb S}_d^+}\int_0^\infty(1-\etr(-r\matr\Theta\matr U))\nu_A(d\matr U,dr) \right)
  \]
  for~$\matr\Theta\in\bar{\mathcal S}_d^+$ with \lvy measure
\begin{equation}\label{eq:nuADef}\begin{split}
  \nu_A(d\matr U,dr)
  &=\int_A\nu(dx,d\matr U,dr)dx\\
  &=\int_A \frac{\exp(-\beta(x,\matr U)r)}{r}dr\alpha(x,\matr U)dx
\end{split}
\end{equation}
  \item For all~$m>0$ and all disjoint measurable~$A_1,\ldots,A_m \subset \mathcal X$, the random matrices~$\matr\Phi(A_1),\ldots,\matr\Phi(A_m)\in\bar{\mathcal S}_d^+$
  are independent and~$\matr\Phi(\sum_jA_j)=\sum_j\matr\Phi(A_j)$.
  \end{enumerate}
\end{thm}
In the situation of Theorem~\ref{th:distOfPhi}, we write~$\matr\Phi \sim \operatorname{GP}_{d\times d}(\alpha,\beta)$
and call~$\matr \Phi$ an \emph{Hpd Gamma process}.
If the process parameter~$\beta$ does not vary among~$\mathcal X$, i.e.~$\beta(x,\cdot)=\tilde\beta(\cdot)$
for a function~$\tilde\beta\colon\bar{\mathbb S}_d^+\to (0,\infty)$,
then~$\matr\Phi$ is called \emph{homogeneous}.
In this case, part~(a) of Theorem~\ref{th:distOfPhi} reveals that~$\matr\Phi(A)\sim\operatorname{Ga}_{d\times d}(\alpha_A, \tilde\beta)$
with the measure~$\alpha_A$ on~$\bar{\mathbb S}_d^+$ being defined as~$\alpha_A(B)=\int_A\alpha(x,B)dx$ for~$B\in\mathbb B(\bar{\mathbb S}_d^+)$.
Part~(b) shows that~$\matr\Phi$ generalizes the notion of a completely random measure 
(i.e.~a random measure~$\Phi$ such that for any finite disjoint collection~$A_1,\ldots,A_m\subset\mathcal X$,
the random variables~$\Phi(A_1),\ldots,\Phi(A_m)$ are independent, see Section~8 in~\cite{kingman1992poisson}).
This is why we will call~$\matr\Phi$ a~\emph{completely random Hpd measure}.
\par
The following result shows that~$\matr\Phi$ obeys a convenient almost surely convergent series 
representation involving iid components.
This is related to the famous stick-breaking representation of
the Dirichlet and the Gamma process~\citep{sethuraman1994constructive,roychowdhuryetal2015} and
will be of great usefulness for later practical applications,
in particular for the implementation of MCMC algorithms.
\begin{thm}\label{th:phiSeries}
Let the assumptions of Theorem~\ref{th:distOfPhi} be fulfilled and assume additionally 
that~$C_\alpha:=\int_{\mathcal X}\alpha(x,\bar{\mathbb S}_d^+)dx$
is finite.
With~$\alpha^*:=\alpha/C_\alpha$, assume that~$\beta(x,\matr U)\geq\beta_0$ holds 
for~$\alpha^*$-almost all~$(x,\matr U)$ and some constant~$\beta_0>0$.
Then
\begin{align*}
  \matr \Phi \overset{\text{a.s.}}{=} \sum_{j \geq 1} \delta_{x_j}r_j\matr U_j, 
  \quad (x_j,\matr U_j) \iid \alpha^*,\\
  r_j = \rho^-(w_j|C_\alpha,\beta(x_j,\matr U_j)),
  \quad w_j = \sum_{i=1}^j v_i,
  \quad v_i \iid \Exp(1),
\end{align*}
where
\[
  \rho^-(w|a,b) = \inf\left\{r>0\colon \rho([r,\infty]|a,b) <w \right\}, \quad  \rho(dr|a,b) = a \frac{\exp(-br)}{r}dr.
\]
\end{thm}
The proof is analogous to Section~3(B) in~\cite{rosinski2001series} and based
on the Interval Theorem and the Marking Theorem for Poisson processes (see~\cite{kingman1992poisson}). 
It may be noted that the assumptions of Theorem~\ref{th:phiSeries} can be generalized
slightly to hold outside a nullset in~$\mathcal X$.
A detailed version of the proof can be found in~\cite{meier}, proof of Lemma~3.13.

\par
For later proofs in the time series applications, further distributional properties
of the Hpd Gamma process are needed, such as full support and lower probability
bounds.
These results may be of independent interest for different applications as well
and can be found in Lemma~\ref{lemma:gpFullSupport} and Lemma~\ref{lemma:gpMass1} in the Appendix.

\section{Spectral Density Inference}\label{sec:bayesProcedure}

We will now illustrate how the Hpd Gamma process can be incorporated in a nonparametric 
prior model for the spectral density matrix~$\matr f$ of a stationary multivariate time series.
To elaborate, let~$I_{j,k}:=((j-1)\pi/k,j\pi/k]$ for~$k>0$ and~$j=1,\ldots,k$.
Denote by~$b(x|j,k-j+1)=\Gamma(k+1)\Gamma(j)^{-1}\Gamma(k-j+1)^{-1}x^{j-1}(1-x)^{k-j}$
for~$0 \leq x \leq 1$ the density of the~$\operatorname{Beta}(j,k-j+1)$ distribution.
We will consider the following \emph{Bernstein-Hpd-Gamma} prior for~$\matr f$:
\begin{equation}\begin{split}\label{eq:hpdGammaPrior}
  \matr f(\omega):=\sum_{j=1}^k\matr\Phi(I_{j,k})b(\omega/\pi|j,k-j+1), \quad 0 \leq \omega \leq \pi, \\
  \matr \Phi \sim \operatorname{GP}_{d\times d}(\alpha,\beta), \quad k \sim p(k),
\end{split}\end{equation}
where~$\alpha$ is a (nonrandom) measure on~$\bar{\mathbb S}_d^+$
and~$\beta\colon\bar{\mathbb S}_d^+\to (0,\infty)$, such that~\eqref{eq:nuAssumption} is fulfilled.
Under this prior, the spectral density is modeled as a Bernstein polynomial mixture
of (random) degree~$k$, where the polynomial mixture weights are Hermitian positive semidefinite matrices
induced by the Hpd Gamma process~$\matr\Phi$.
The latter is defined on~$\mathcal X = [0,\pi]$ and the prior probability mass function~$p(k)$ of the
polynomial degree is fully supported on~$\mathbb N$.
A common choice (which is motivated from asymptotic considerations, as outlined
in the upcoming Section~\ref{sec:asymptotics}), is~$p(k)\propto \exp(-ck\log k)$ for~$k \in \mathbb N$,
where~$c$ is a positive constant.
This Bernstein polynomial mixture approach is inspired by existing methods for univariate 
spectral density inference~\citep{choudhuri} and density estimation~\citep{petrone1999random}
and relies on the approximation properties of Bernstein polynomials (see e.g.~Section~1.6
in~\cite{lorentz2012bernstein} or Appendix~E in~\cite{ghosal2017fundamentals}).

\par
The model specification is completed by employing Whittle's likelihood~$P_W^n$ for the spectral density.
It is defined in terms of the Fourier coefficients
\begin{equation}\label{eq:fourierCoef}
  \vectt Z_j := \frac{1}{\sqrt n} \sum_{t=1}^n \vect Z_t \exp(-it\omega_j),
  \quad \omega_j:=\frac{2\pi j}{n}, 
  \quad j=0,\ldots,\left\lfloor \frac{n}{2} \right\rfloor,
\end{equation}
of the data~$\vect Z_1,\ldots,\vect Z_n$.
Asymptotically, the Fourier coefficients are independent and normally distributed,
a result that is well known and holds true in quite great generality,
beyond Gaussian time series, see~\cite{hannan}, p.224.
Motivated from this asymptotic result, Whittle's likelihood is a pseudo likelihood that mirrors
the asymptotic distribution of the Fourier coefficients.
With~$N:=\lceil n/2 \rceil -1$, it is defined in terms of the Lebesgue density
\begin{equation}\label{eq:whittle}
  p_W^n(\vectt z_1,\ldots,\vectt z_N|\matr f)
  = \prod_{j=1}^N \frac{1}{\pi^d|2\pi\matr f(\omega_j)|}
  \exp\left( -\frac{1}{2\pi}\vectt z_j^* \matr f(\omega_j)^{-1}\vectt z_j \right),
\end{equation}
for~$\vectt z_1,\ldots,\vectt z_N\in\mathbb C^d$.
The coefficients corresponding to the Fourier frequencies~$\omega=0$
and~$\omega=\pi$ (the latter occurring for~$n$ even) are omitted in the likelihood, 
since they represent the mean and alternating mean resp.,
which are both typically treated separately from the autocovariance structure.
By Bayes' Theorem, this leads to the posterior 
distribution~$P_W^n(\matr f|\vect Z_1,\ldots,\vect Z_n) \propto P_W^n(\vect Z_1,\ldots,\vect Z_n|\matr f)P(\matr f)$.

\par
Since the posterior is not tractable analytically, we employ a
Markov Chain Monte Carlo (MCMC) algorithm to draw random samples from it.
To this end, we approximate the infinite series representation of~$\matr \Phi$ from Theorem~\ref{th:phiSeries}
by a finite sum as~$\matr\Phi \approx \sum_{j=1}^L \delta_{x_j}r_j\matr U_j$ for some large integer~$L$,
serving as a truncation parameter.
In practice, the choice of~$L$ should depend on the sample size~$n$.
In chapter~\ref{sec:illustration} below we choose~$L=\max\{20,n^{1/3}\}$, which is the same
value that has also been used by~\cite{choudhuri} and~\cite{meier}, Section~3.4.3 and Section~5.2.
With this approximative representation of~$\matr\Phi$, the spectral density~$\matr f$
is parametrized by the finite-dimensional vector
\begin{equation}\label{eq:thetaParametrization}
  \vect\Theta_{\matr f} = (k,x_1,\ldots,x_L,\matr U_1,\ldots,\matr U_L,r_1,\ldots,r_L).
\end{equation}
Posterior samples can be drawn with a Metropolis-within-Gibbs sampler,
which is discussed in more detail in Appendix~\ref{sec:app:mcmc}.

\section{Asymptotic Properties}\label{sec:asymptotics}
We will establish $\mathbb L^1$-consistency (Theorem~\ref{th:consistency} in Section~\ref{sec:consistency})
as well as uniform consistency (Theorem~\ref{th:consietencyUniform} in Section~\ref{sec:consistencyUniform})
and Hellinger contraction rates (Theorem~\ref{th:contraction} in Section~\ref{sec:contraction}) for the posterior of~$\matr f$.
As a first important observation, we will derive that Whittle's likelihood
and the full Gaussian likelihood are mutually contiguous (Theorem~\ref{th:contiguity} in Section~\ref{sec:contiguity}).
Our considerations are based on the following assumption on the \emph{true} spectral density~$\matr f_0$.
\begin{assumption}\label{ass:f0}
Let~$\{\vect Z_t \colon \vect t \in \mathbb Z\}$ be a Gaussian stationary time
series in~$\mathbb R^d$ with spectral density~$\matr f_0$ fulfilling:
\begin{enumerate}[label=(\alph*)]
\item
There exist constants~$b_0,b_1>0$ such that
\[
  \lambda_{1}(\matr f_0(\omega)) \geq b_0, \quad
  \lambda_{d}(\matr f_0(\omega)) \leq b_1, \quad
  0 \leq \omega \leq \pi,
\]
where~$\lambda_1,\lambda_d$ denote the smallest and largest eigenvalue respectively.
\item
There exists~$a>1$ such that~$\matr\Gamma_0(h)=\int_0^{2\pi}\matr f_0(\omega)\exp(ih\omega)d\omega$
fulfills
\[
  \sum_{h\in\mathbb Z} \|\matr\Gamma_0(h)\| |h|^a < \infty.
\]
\end{enumerate}
\end{assumption}
Part~(a) of Assumption~\ref{ass:f0} has a statistical interpretation in terms
of decay of linear dependence coefficients, at least for the univariate case (see Theorem~1.6 in~\cite{bradley2002positive}).
Part~(b) can be seen as a regularity condition, since it implies~$\matr f$ 
to be continuously differentiable with derivative being H\"older of order~$a-1>0$.

\subsection{Contiguity}\label{sec:contiguity}
We will consider the following version~$\tilde P_W^n=\tilde P_W^n(\cdot|\matr f)$ of Whittle's likelihood, where
the coefficients corresponding to the Fourier frequencies~$\omega=0$
and~$\omega=\pi$ are taken into account:
\[
  \tilde p_W^n(\vectt z_0,\ldots, \vectt z_{\lfloor n/2 \rfloor} | \matr f)
  = p_{0,n}(\vectt z_0|\matr f) p_W^n(\vectt z_1,\ldots,\vectt z_N|\matr f)p_{n/2,n}(\vectt z_{n/2}),
\]
with~$p_W^n$ from~\eqref{eq:whittle}
and~$p_{0,n}(\cdot | \matr f)$ being the density of the~$N_d(\vect 0,2\pi\matr f(0))$ distribution
and~$p_{n/2,n}(\cdot | \matr f)$ the density of the~$N_d(\vect 0,2\pi\matr f(\pi))$ distribution,
the latter being included if and only if~$n$ is even.
The joint distribution of~$\vectt Z_0,\ldots, \vectt Z_{\lfloor n/2 \rfloor}$ will be denoted by~$\tilde P^n=\tilde P^n(\cdot|\matr f)$.
The following result generalizes the findings from~\cite{choudhuriContiguity} to the multivariate case.
Recall that two sequences~$(P_n)$ and~$(Q_n)$ of measures on measurable spaces~$\mathcal X_n$ 
are called \emph{mutually contiguous}, if for every sequence~$(A_n)$ of measurable sets it holds: 
$P_n(A_n) \to 0$ if and only if~$Q_n(A_n) \to 0$.
\begin{thm}\label{th:contiguity}
Let Assumption~\ref{ass:f0} be fulfilled.
Then~$\tilde P_W^n$ and~$\tilde P^n$ are mutually contiguous.
\end{thm}
The result carries over to the version~$P_W^n$ of Whittle's likelihood from~\eqref{eq:whittle},
as formulated in the following Corollary.
\begin{corollary}\label{cor:contiguity}
Let Assumption~\ref{ass:f0} be fulfilled
and denote by~$P^n=P^n(\cdot|\matr f)$ the joint distribution of~$\vectt Z_1,\ldots,\vectt Z_N$.
Then~$P_W^n$ and~$P^n$ are mutually contiguous.
\end{corollary}
The proof of contiguity relies on the Gaussianity assumption and it can
be shown that the result does not extend beyond Gaussianity in general
(see the discussion in Section~4.3 in~\cite{meier}.
It is conjectured that posterior consistency and contraction rates
may be valid for non-Gaussian time series as well, however,
since the proof relies on the contiguity result, it is unclear how this can be shown.

\subsection{$\mathbb L^1$-consistency}\label{sec:consistency}
Consistency for the spectral density of a univariate time series under a 
Bernstein-Dirichlet prior in the~$\mathbb L^1$-topology has been derived in~\cite{choudhuri}.
The authors relied on the assumption that the normalizing constant of the spectral density is known a priori.
Our method of proof is slightly different and does not rely on such an assumption, 
but requires (the eigenvalues of) the spectral density to be uniformly bounded a priori.
To elaborate, denote by~$\mathcal C$ the set of continuous~$\bar{\mathcal S}_d^+$-valued functions on~$[0,\pi]$.
We endow~$\mathcal C$ with the~$\sigma$-algebra induced by the maximum Frobenius norm~$\|\matr f\|_{F,\infty}:=\max_{0\leq\omega\leq\pi}\|\matr f(\omega)\|$, where we will also consider the Frobenius~$\mathbb L^p$ norms~$\|\matr f\|_{F,p}:=(\int\|\matr f(\omega)\|^pd\omega)^{1/p}$
for~$p> 0$.
Let~$0 \leq \tau_0<\tau_1 \leq \infty$ and consider the set
\begin{equation}\label{eq:truncationSet}
  \mathcal C_{\tau_0,\tau_1} := \{ \matr f \in \mathcal C \colon 
  \lambda_{1}(\matr f(\omega)) > \tau_0, 
  \lambda_{d}(\matr f(\omega)) < \tau_1
  \text{ for all } 0 \leq \omega \leq \pi \}.
\end{equation}
Denoting by~$P$ the Bernstein-Hpd-Gamma prior~\eqref{eq:hpdGammaPrior} on~$\matr f \in \mathcal C$,
let~$P_{\tau_0,\tau_1}$ be the restriction of~$P$ to~$\mathcal C_{\tau_0,\tau_1}$.
The pseudo-posterior when updating~$P_{\tau_0,\tau_1}$ with Whittle's
likelihood~$P_W^n$ from~\eqref{eq:whittle} will be denoted by~$P_{W;\tau_0,\tau_1}^n(\matr f | \vect Z_1,\ldots,\vect Z_n)$.
We will make the following assumptions on~$\matr \Phi \sim \operatorname{GP}_{d\times d}(\alpha,\beta)$
and on the prior probability mass function~$p(k)$ of the Bernstein polynomial degree~$k$.
\begin{assumption}\label{ass:gp1}
\begin{enumerate}[label=(\alph*)]
\item
  The measure~$\alpha(x,\cdot)$ is of full support on~$\bar{\mathbb S}_d^+$
  for almost all~$x \in [0,\pi]$.
\item
  It holds~$\sup_{0 \leq x \leq \pi, \matr U \in \bar{\mathbb S}_d^+}\beta(x,\matr U)<\infty$.
\item
  There exist positive constant~$c,C$ such that~$0<p(k)\leq C\exp(-ck\log k)$ holds for all~$k \in \mathbb N$.
\end{enumerate}
\end{assumption}
Part~(a) and~(b) of Assumption~\ref{ass:gp1} are needed to ensure that 
prior probability mass is allocated in arbitrarily small neighborhoods
of the true spectral density~$\matr f_0$.
In fact, it ensures that the Hpd Gamma process~$\matr\Phi$ assigns positive
mass to neighborhoods of the true spectral measure~$\matr F_0 := \int \matr f_0(\omega)d\omega$.
Part~(c) of Assumption~\ref{ass:gp1} states full prior support for~$k$
and poses a condition on the decay of the tail of the prior.
This is similar to frequentist tuning parameters that control the smoothness of the estimator,
and is the same condition as in the univariate case, see~\cite{choudhuri}.
Now we can formulate the first main result of this section, stating~$\mathbb L^1$
consistency of the posterior for Gaussian time series.
\begin{thm}\label{th:consistency}
Let~$\{ \vect Z_t \}$ be a stationary Gaussian time series in~$\mathbb R^d$
with true spectral density~$\matr f_0$ fulfilling Assumption~\ref{ass:f0}.
Let~$\tau\in (b_1,\infty)$. 
Consider the prior~$P_{0,\tau}$ on~$\matr f$,
with Bernstein-Hpd-Gamma prior fulfilling Assumption~\ref{ass:gp1}.
Then for all~$\varepsilon>0$ it holds
\[
  P_{W;0,\tau}^n\left( \left\{ \matr f \colon \|\matr f - \matr f_0\|_{F,1}<\varepsilon \right\} \bigm| \vect Z_1,\ldots,\vect Z_n \right) \to 1, \quad \text{in } P^n(\cdot,\matr f_0) \text{ probability},
\]
with~$P^n(\cdot,\matr f_0)$ denoting the joint distribution of~$\vectt Z_1,\ldots,\vectt Z_N$ from~\eqref{eq:fourierCoef}.
\end{thm}

\subsection{Uniform consistency}\label{sec:consistencyUniform}
The posterior consistency result can be strengthened from the~$\mathbb L^1$ topology
to the uniform topology, under strengthened assumptions of the prior.
The idea is to restrict the prior not only in the~$\|\cdot\|_{F,\infty}$
topology, but also introduce a uniform bound on the Lipschitz constants.
To elaborate, consider the following Lipschitz norm on the space~$\mathcal C^1$
of continuously differentiable Hpd matrix valued functions on~$[0,2\pi]$:
\[
  \| \matr f \|_{L} := \|\matr f\|_{F,\infty} + \sup_{\omega_1 \neq \omega_2} \frac{\| \matr f(\omega_1)-\matr f(\omega_1) \|}{|\omega_1-\omega_2|}.
\]
For~$\tau>0$, let
\begin{equation}\label{eq:lipschitzTruncationSet}
  \tilde{\mathcal C}_\tau := \{ \matr f \in \mathcal C^1 \colon \|\matr f\|_L \leq \tau \}
\end{equation}
and denote by~$\tilde P_\tau$ the restriction of the Bernstein-Hpd-Gamma prior~\eqref{eq:hpdGammaPrior}
to~$\tilde{\mathcal C}_\tau$.
Denote by~$\tilde P_{W,\tau}^n$ the pseudo-posterior distribution
for~$\matr f$ when employing~$\tilde P_\tau$ in conjunction with Whittle's likelihood~$P_W^n$ from~\eqref{eq:whittle}
\begin{thm}\label{th:consietencyUniform}
Let~$\{ \vect Z_t \}$ be a stationary Gaussian time series in~$\mathbb R^d$
with true spectral density~$\matr f_0$ fulfilling Assumption~\ref{ass:f0} such that~$\| \matr f_0 \|_{L} \leq b_1$.
Let~$\tau\in (b_1,\infty)$. 
Consider the prior~$\tilde P_\tau$ on~$\matr f$,
with Bernstein-Hpd-Gamma prior fulfilling Assumption~\ref{ass:gp1}.
Then for all~$\varepsilon>0$ it holds
\[
  \tilde P_{W,\tau}^n\left( \left\{ \matr f \colon \| \matr f-\matr f_0 \|_{F,\infty} < \varepsilon \right\} \bigm| \vect Z_1,\ldots,\vect Z_n \right) \to 1, \quad \text{in } P^n(\cdot,\matr f_0) \text{ probability}.
\]
\end{thm}

\subsection{Posterior contraction rates}\label{sec:contraction}
Consider
the following average squared Hellinger distance~$d_{n,H}^2(\matr f_0,\matr f)$
between two spectral density matrices~$\matr f_0, \matr f$:
\begin{equation}\label{eq:hellingerFDef}
  d_{n,H}^2(\matr f_0,\matr f)
  := \frac{1}{N} \sum_{j=1}^N d_H^2\big( p_{j,N}(\cdot|\matr f_0), p_{j,N}(\cdot|\matr f) \big),
\end{equation}
where~$d_H^2(p,q) = 1-\int\sqrt{p(x)q(x)}dx$ 
denotes the squared Hellinger distance between two probability densities~$p,q$ and
\begin{equation}\label{eq:pjn}
  p_{j,N}( \vect z | \matr f_0)= \frac{1}{\pi^{d}|2\pi\matr f(\omega_j)|}\exp\left(-\vect z^* (2\pi\matr f(\omega_j))^{-1}\vect z\right), 
  \quad \vect z \in \mathbb C^d,
\end{equation}
denotes the density of the complex multivariate normal~$CN_d(\vect 0, 2\pi\matr f(\omega_j))$ 
distribution.
Observe that~$p_{j,N}$ corresponds to
the distribution of the Fourier coefficient~$\vectt Z_j$ under Whittle's 
likelihood~\eqref{eq:whittle}.
It can readily be seen that~$d_{n,H}$ is a semimetric (i.e.~a symmetric nonnegative
function satisfying the triangle inequality, but possibly lacking the identity of indiscernibles) 
on~$\mathcal C$. 
We will need the following assumptions on the prior.
\begin{assumption}\label{ass:gp2}
\begin{enumerate}[label=(\alph*)]
\item
  There exists~$g \colon [0,\pi] \times \bar{\mathbb S}_d^+ \to (0,\infty)$ 
  and~$g_0,g_1>0$ such that it holds~$\alpha(x,d\matr U)=g(x,\matr U)d\matr U$
  and~$g_0 \leq g(x,\matr U) \leq g_1$ for all~$0 \leq x \leq \pi, \matr U \in \bar{\mathbb S}_d^+$.
\item
  There exist~$\beta_0,\beta_1>0$ such that~$\beta_0 \leq \beta(x,\matr U) \leq \beta_1$ holds for all~$x,\matr U$.
\item
  It holds~$A_1\exp(-\kappa_1 k\log k) \leq p(k) \leq A_2\exp(-\kappa_2k)$
  for all~$k \in \mathbb N$
  and constants~$A_1,A_2,\kappa_1,\kappa_2>0$.
\end{enumerate}
\end{assumption}
Part~(a) and~(b) of Assumption~\ref{ass:gp2} are stronger than part~(a) and~(b)
of Assumption~\ref{ass:gp1}.
They ensure not only positive prior mass around~$\matr f_0$, but also enable
the derivation of lower bounds for the prior probability mass of small neighborhoods.
As an example, the prior choice in Section~\ref{sec:illustration}
fulfills these conditions.
Further examples are discussed in~\cite{meier}, Lemma~3.9.
Part~(c) is a condition on the prior tail of~$k$.
In contrast to part~(c) of Assumption~\ref{ass:gp2}, the decay is
bounded not only from below but also from above.
The following theorem establishes contraction rates for the posterior of~$\matr f$.
The rates coincide with the rates that are known for the univariate case, see~\cite{ghosal2007convergence}
and Example~9.19 in~\cite{ghosal2017fundamentals}.
\begin{thm}\label{th:contraction}
Let~$\{ \vect Z_t \}$ be a stationary Gaussian time series in~$\mathbb R^d$
with true spectral density~$\matr f_0$ fulfilling Assumption~\ref{ass:f0} with~$1<a\leq 2$.
Let~$0<\tau_0<b_0<b_1<\tau_1<\infty$ and let the prior on~$\matr f$ be given by~$P_{\tau_0,\tau_1}$,
with Bernstein-Hpd-Gamma prior fulfilling Assumption~\ref{ass:gp2}.
\par
Then with~$\varepsilon_n = n^{-a/(2+2a)}(\log n)^{(1+2a)/(2+2a)}$ it holds
\[
  P_{W;\tau_0,\tau_1}^n\left( \left\{ \matr f \colon d_{n,H}(\matr f_0,\matr f) < M_n\varepsilon_n \right\} \bigm| \vect Z_1,\ldots,\vect Z_n \right) \to 1
\]
in~$P^n(\cdot,\matr f_0$ probability,
for every positive sequence~$M_n$ with~$M_n\to\infty$.
\end{thm}

\begin{remark}
In Theorem~\ref{th:contraction}, the root average squared Hellinger distance~$d_{n,H}$
is used for technical reasons, as the metric in consideration needs to yield
exponentially powerful and uniformly exponentially consistent tests
for this proof to be valid.
It is known that such tests always exist for the Hellinger topology,
while deriving such tests for different topologies may constitute a difficult task.
Under additional prior assumptions, the Hellinger contraction rates
from Theorem~\ref{th:contraction} can be used to establish rates in different
topologies.
As an example, if the prior is restricted to a Lipschitz class with uniformly
bounded Lipschitz constant, then Lemma~\ref{lemma:lowerRootHellingerBound}
reveals that~$\varepsilon_n$ is also a contraction rate with respect to the~$\|\cdot\|_{F,1}$ norm
and that~$\varepsilon_n^{2/3}$ is a contraction rate in the~$\|\cdot\|_{F,\infty}$ norm.
It is, however, not clear whether these rates could be improved (or the prior assumptions relaxed) by a different proof technique.
\end{remark}

\section{Illustration}\label{sec:illustration}
We will compare the performance of our proposed method
with a parametric VAR model, where we will consider both simulated (Section~\ref{sec:simulationStudy}) 
and a real data example (Section~\ref{sec:soiData}).
An implementation of all procedures presented below is included in
the \verb|R| package \verb|beyondWhittle|, which is available on CRAN, see~\cite{beyondWhittlePackage}.
As Bayes estimates, we will consider the pointwise posterior median function~$\hat{\matr f}_0$,
with~$\hat{\matr f}_0(\omega)$ consisting of the posterior median of
real-and imaginary parts of the components of~$\matr f(\omega)$ for~$0\leq\omega\leq\pi$.
We will compute pointwise~90\% posterior credibility regions as the area between
the pointwise~0.05 and~0.95 quantiles of the real and imaginary parts of the components of~$\matr f$.
We will also consider \emph{uniform credibility regions}.
To elaborate, assume that we have a posterior sample~$\matr f^{(1)},\ldots,\matr f^{(M)}$
at hand, as e.g.~generated with the algorithm described in Appendix~\ref{sec:app:mcmc}.
Denote by~$\mathcal H \colon \mathcal S_d^+ \to \mathbb R^{d^2}$ the transformation
that maps each Hermitian positive definite matrix~$\matr A=(a_{rs})_{r,s=1}^d$ to a vector~$\mathcal H\matr A$
consisting of the logarithmized diagonal elements~$\log a_{11},\ldots,\log a_{dd}$ and the 
(non-logarithmized) real and imaginary parts of the entries~$\{ a_{rs} \colon r < s \}$ above the diagonal.
For the transformed versions~$ \vect {h}\ups{j} = ( h\ups{j}_1,\ldots, h\ups{j}_{d^2}) := \mathcal H\matr f\ups{j}$ for~$j=1,\ldots,M$,
denote the pointwise sample median function by~$\vect {\hat h}:=( \hat h_1,\ldots, \hat h_{d^2})$.
Let~$\vect{\hat\sigma}:=(\hat\sigma_1,\ldots,\hat\sigma_{d^2})$ with~$\hat\sigma_r(\omega)$
being the median absolute deviation of~$\{h\ups{1}_r(\omega),\ldots,h\ups{M}_r(\omega) \}$.
Let~$C_{0.9}$ be the smallest positive number such that
\[
  \frac{1}{M}\sum_{j=1}^M \indi 
  \left\{ 
    \max_{\substack{0\leq \omega \leq \pi \\ r=1,\ldots,d^2}} 
    \frac{\left| h\ups{j}_r(\omega)-\hat h_r(\omega) \right|}{\hat\sigma_r(\omega)}
    \leq C_{0.9}
  \right\} \geq 0.9.
\]
Let~$\vect{\hat h}\upss{0.05}:=\vect{\hat h}-C_{0.9}\vect{\hat\sigma}$
and~$\vect{\hat h}\upss{0.95}:=\vect{\hat h}+C_{0.9}\vect{\hat\sigma}$
and~$\tilde{\matr f}_0\upss{0.05}:=\mathcal H^{-1}\vect{\hat h}\upss{0.05}$
as well as~$\tilde{\matr f}_0\upss{0.95}:=\mathcal H^{-1}\vect{\hat h}\upss{0.95}$.
Then 
\[
  \mathcal C_{\operatorname{uni}} (\omega|0.9) := \left\{ t\tilde{\matr f}_0\upss{0.05}(\omega) + (1-t)\tilde{\matr f}_0\upss{0.95}(\omega) \colon 0 \leq t \leq 1 \right\}, \quad 0 \leq \omega \leq \pi
\]
is called a uniform 90\% credibility region.
By construction, it holds that~$\matr f \in \mathcal C_{\operatorname{uni}} (\cdot|0.9)$ with
(empirical) posterior probability of at least~90\%.

\subsubsection*{Prior choice}
The parameters of the Hpd Gamma process in the 
Bernstein-Hpd-Gamma prior~\eqref{eq:hpdGammaPrior} 
are chosen such that the mean is large and the covariance is large,
and proportional to the identity matrix.
This is done in order to achieve a prior for~$\matr f$ that is
vague, homogeneous and isotropic (i.e.~not preferring any particular
directions in the function space).
To elaborate, we choose the process parameters~$\alpha(x,d\matr U)= \alpha_0(d\matr U)$
and~$\beta(x,d\matr U)\equiv \beta_0$ of~$\operatorname{GP}_{d\times d}(\alpha,\beta)$
as~$\alpha_0(d\matr U)=2d\matr U$ and~$\beta_0=10^4$.
As mentioned in Remark~\ref{rem:aGamma}, we have that for~$\matr Z\sim\operatorname{Ga}_{d\times d}(\alpha_0,\beta_0)$
it holds~$\E\matr Z=10^4\matr I_d$ with~$\Cov\matr Z:=\E[\matr Z\otimes \matr Z]-(\E\matr Z\otimes\E\matr Z)$ 
being component-wise proportional to~$10^8\matr I_d$.

\par
To achieve a more stable mixture behavior at the left and right boundary of~$[0,\pi]$, 
the beta densities~$b(\cdot|j,k-j+1)$ in~\eqref{eq:hpdGammaPrior} 
are replaced by their \emph{truncated} and dilated counterparts.
These are defined, for~$0<\xi_l<\xi_r<1$, as
\begin{equation}\label{eq:truncatedBernstein}
  b_{\xi_l}^{\xi_r}(x|j,k-j+1) := b(\xi_l + x(\xi_r - \xi_l)|j,k-j+1), \quad 0 \leq x \leq 1
\end{equation}
and have the advantage that the polynomial mixture at~$x=0$ (and~$x=1$ resp.) is not
only determined by~$b(\cdot|1,k)$ (and~$b(\cdot|k,1)$ resp.), but other basis functions
are also incorporated.
See Section~5.1 in~\cite{meier} for a more detailed discussion of this matter.
We will employ the truncated Bernstein polynomial basis from~\eqref{eq:truncatedBernstein} with~$\xi_l=0.1$
and~$\xi_r=0.9$.

\par
The prior of the polynomial degree~$k$ is~$p(k)\propto\exp(-0.01k\log k)$.
The values of~$k$ are 
thresholded at~$k_{\max}=500$, which was found to be large enough in preliminary pilot runs,
for the sake of computational speed-up (see Appendix~\ref{sec:app:mcmc}).
Posterior samples are obtained with the MCMC algorithm from Appendix~\ref{sec:app:mcmc}.
Each Markov chain is run for a total of~80,000 iterations,
where the first~30,000 iterations are discarded as burn-in period
and the remaining~50,000 iterations are thinned by a factor of~5,
yielding a posterior sample size of~10,000.
This procedure will be referred to as the \emph{NP procedure} in the following,
where NP stands for \underline{n}on\underline{p}arametric method.

\subsection{Simulated data}\label{sec:simulationStudy}
We consider simulated data drawn from the following bivariate VAR(2) model
which has been considered in~\cite{rosen2007automatic}:
\begin{align}\label{eq:var2Model}
  &\vect Z_t = 
  \begin{pmatrix} 0.5 & 0 \\ 0 & -0.3 \end{pmatrix} \vect Z_{t-1} +
  \begin{pmatrix} 0 & 0 \\ 0 & -0.5 \end{pmatrix} \vect Z_{t-2} + \vect e_t, \\
  &\quad \{\vect e_t\} \iid \operatorname{WN}_d(\vect 0,\matr\Sigma_{\operatorname{VAR}}),
  \quad \matr\Sigma_{\operatorname{VAR}}= \begin{pmatrix} 1 & 0.9 \\ 0.9 & 1 \end{pmatrix},\nonumber
\end{align}
where~$\operatorname{WN}_d(\vect\mu,\matr\Sigma)$ 
denotes~$d$-dimensional White Noise with mean~$\vect\mu\in\mathbb R^d$ and
covariance matrix~$\matr\Sigma\in\mathcal S_d^+$,
which is generated from a standard White Noise~$\{\vectt e_t\}$ 
with~$\vectt e_t = (\vectt e_{t,1},\ldots,\vectt e_{t,d})$
and~$\vectt e_{t,1},\ldots,\vectt e_{t,d} \iid \operatorname{WN}(0,1)$
as~$\vect e_t := \matr\Sigma^{1/2}\vectt e_t+\vect\mu$.
Furthermore, we draw from the following bivariate VMA(1) model:
\begin{align}\label{eq:vma1Model}
  &\vect Z_t = \vect e_t +
  \begin{pmatrix} -0.75 & 0.5 \\ 0.5 & 0.75  \end{pmatrix} \vect e_{t-1},\\
  &\quad \{\vect e_t\} \sim \operatorname{WN}_d(\vect 0,\matr\Sigma_{\operatorname{VMA}}),
  \quad \matr\Sigma_{\operatorname{VMA}} = \begin{pmatrix} 1 & 0.5 \\ 0.5 & 1 \end{pmatrix}.
\end{align}
We will consider normally distributed innovations (i.e.~$\vect e_{t,j} \iid N(0,1)$).
Furthermore, we will also consider Student-t distributed innovations with~$\nu=4$
degrees of freedom and centered exponential innovations (i.e.~$\vect e_{t,j}+1 \iid \operatorname{Exp}(1)$).
For~$N=\lceil n/2 \rceil -1$, we compare the~$\mathbb L^1$-error~$\| \hat{\matr f}_0-\matr f_0 \|_1 :=N^{-1} \sum_{j=1}^{N} \| \hat{\matr f}_0(\omega_j)-\matr f_0(\omega_j) \|$
of the pointwise posterior median function~$\hat{\matr f}_0$.
This is close to~$\int \| \hat{\matr f}_0(\omega)-\matr f_0(\omega) \|d\omega$ for~$N$ large.
The~$\mathbb L^2$-error is defined analogously.
As a parametric comparison model, we employ a Gaussian Vector Autoregession (VAR)
\[
  \vect Z_t = \sum_{j=1}^p \matr B_j\vect Z_{t-j} + \vect e_t,
  \quad \{\vect e_t\} \iid N_d(\vect 0,\matr \Sigma)
\]
with \emph{Normal-Inverse-Wishart} prior~$\matr \Sigma \sim \operatorname{Wish}_{d\times d}^{-1}(10^{-4},10^{-4}\matr I_d)$
(see Section~3.4 in \cite{gupta})
and~$\operatorname{vec}(\matr B_1,\ldots,\matr B_p)\sim N_{pd^2}(\vect 0,10^4\matr I_{pd^2})$,
where~$\operatorname{vec}$ denotes the vectorization operator that stacks
all columns of a matrix below each other.
The order~$p$ is determined in a preliminary model selection step based on
Akaike's Information Criterion~\citep{akaike1974new}.
To draw posterior samples, the Gibbs sampling algorithm from Section~2.2.3 in~\cite{koop2010bayesian}
is employed.
The Markov Chain lengths and posterior sample size are chosen as for the NP procedure.
We will refer to this method as the \emph{VAR procedure} in the following.
We consider~$M=500$
independent realizations of models~\eqref{eq:var2Model} as well as~\eqref{eq:vma1Model}
for each length~$n=256,512,1024$.

\par
The results are shown in Table~\ref{tab:varma}.
\begin{table}
\centering
\begin{subtable}[c]{\textwidth}
\centering
\fbox{
\begin{tabular}{lrrlrrlrr}
& \multicolumn{8}{c}{VAR(2) data} \\
& \multicolumn{2}{c}{$n=256$} && \multicolumn{2}{c}{$n=512$} && \multicolumn{2}{c}{$n=1024$}  \\
\cline{2-3}\cline{5-6}\cline{8-9}
& NP & VAR && NP & VAR && NP & VAR \\
$\mathbb L^1$-error & 0.105 & 0.071    && 0.081 & 0.050     && 0.064 & 0.034 \\
$\mathbb L^2$-error & 0.133 & 0.094    && 0.107 & 0.066     && 0.085 & 0.045 \\
~\\[-1ex]
& \multicolumn{8}{c}{VMA(1) data} \\
& \multicolumn{2}{c}{$n=256$} && \multicolumn{2}{c}{$n=512$} && \multicolumn{2}{c}{$n=1024$}  \\
\cline{2-3}\cline{5-6}\cline{8-9}
& NP & VAR && NP & VAR && NP & VAR \\
$\mathbb L^1$-error & 0.095 & 0.155    && 0.070 & 0.121    && 0.053 & 0.091 \\
$\mathbb L^2$-error & 0.113 & 0.187    && 0.084 & 0.144    && 0.064 & 0.108 \\
\end{tabular}
}
\subcaption{}
\end{subtable}
\begin{subtable}[c]{\textwidth}
\centering
\fbox{
\begin{tabular}{lrrlrrlrr}
& \multicolumn{8}{c}{VAR(2) data} \\
& \multicolumn{2}{c}{$n=256$} && \multicolumn{2}{c}{$n=512$} && \multicolumn{2}{c}{$n=1024$}  \\
\cline{2-3}\cline{5-6}\cline{8-9}
& NP & VAR && NP & VAR && NP & VAR \\
$\mathbb L^1$-error & 0.115 & 0.089    && 0.092 & 0.067     && 0.071 & 0.047 \\
$\mathbb L^2$-error & 0.146 & 0.113    && 0.119 & 0.085     && 0.094 & 0.059 \\
~\\[-1ex]
& \multicolumn{8}{c}{VMA(1) data} \\
& \multicolumn{2}{c}{$n=256$} && \multicolumn{2}{c}{$n=512$} && \multicolumn{2}{c}{$n=1024$}  \\
\cline{2-3}\cline{5-6}\cline{8-9}
& NP & VAR && NP & VAR && NP & VAR \\
$\mathbb L^1$-error & 0.112 & 0.168    && 0.088 & 0.135     && 0.067 & 0.101 \\
$\mathbb L^2$-error & 0.129 & 0.201    && 0.103 & 0.159     && 0.078 & 0.118 \\
\end{tabular}
}
\subcaption{}
\end{subtable}
\begin{subtable}[c]{\textwidth}
\centering
\fbox{
\begin{tabular}{lrrlrrlrr}
& \multicolumn{8}{c}{VAR(2) data} \\
& \multicolumn{2}{c}{$n=256$} && \multicolumn{2}{c}{$n=512$} && \multicolumn{2}{c}{$n=1024$}  \\
\cline{2-3}\cline{5-6}\cline{8-9}
& NP & VAR && NP & VAR && NP & VAR \\
$\mathbb L^1$-error & 0.110 & 0.085    && 0.086 & 0.061     && 0.067 & 0.042 \\
$\mathbb L^2$-error & 0.139 & 0.110    && 0.112 & 0.079     && 0.088 & 0.054 \\
~\\[-1ex]
& \multicolumn{8}{c}{VMA(1) data} \\
& \multicolumn{2}{c}{$n=256$} && \multicolumn{2}{c}{$n=512$} && \multicolumn{2}{c}{$n=1024$}  \\
\cline{2-3}\cline{5-6}\cline{8-9}
& NP & VAR && NP & VAR && NP & VAR \\
$\mathbb L^1$-error & 0.107 & 0.164    && 0.080 & 0.127     && 0.061 & 0.097 \\
$\mathbb L^2$-error & 0.124 & 0.198    && 0.094 & 0.152     && 0.072 & 0.115 \\
\end{tabular}
}
\subcaption{}
\end{subtable}
\caption{$\mathbb L^1$- and $\mathbb L^2$-error
of NP procedure and VAR procedure for VAR(2) and VMA(1) data
and~(a) Gaussian innovations, Student(t) innovations and~(c)
centered exponential innovations.}\label{tab:varma}
\end{table}
As for the~$\mathbb L^1$- and~$\mathbb L^2$-error, it can be seen that the VAR 
procedure outperforms the NP procedure for VAR(2) data,
which can be expected since in this case the parametric model is well-specified.
For VMA(1) data however, the NP procedure yields better results
illustrating the benefit of employing a nonparametric procedure in comparison to a parametric method
by being much less susceptible to misspecification.

The empirical coverage of uniform~90\% credibility regions is smaller for
the NP procedure than for the VAR procedure in all examples (e.g.~0.37 vs.~0.90 
for VAR(2) data or~0.44 vs.~0.98 for VMA(1) data with~$n=512$).
With view on the discussion in~\cite{szabo2015frequentist},
we conjecture that this property is due to the usage of Bernstein polynomials, which are known to have suboptimal
approximation rates and tend to produce over-smoothed results 
(see Appendix~E in~\cite{ghosal2017fundamentals}
and the simulation study in~\cite{Edwards2018}).

Possible alternative approaches include the usage of a different
polynomial basis (e.g.~B-splines as in~\cite{Edwards2018})
or to employ a parametric working model that increases the prior flexibility,
as considered in~\cite{kirchMeyer} for the univariate case.

\subsection{Southern Oscillation Index}\label{sec:soiData}
We analyze the Southern Oscillation Index (SOI) and Recruitment series 
that have been analyzed in~\cite{rosen2007automatic} and~\cite{shumway2010time}.
Both time series are available
as datasets \texttt{soi} and \texttt{rec} in the \texttt{R} package \texttt{astsa}~\citep{astsaR}.
They consist of monthly data for 452 months.
The SOI is defined as the normalized difference in air pressure
between Tahiti (French Polynesia) and Darwin (Northern Territory, Australia). It constitutes
a key indicator for warming or cooling effects of the central and eastern Pacific ocean known
as El Ni\~{n}o and La Ni\~{n}a, see~\cite{soiGlossary,soiGlossary2}.
The Recruitment Series consists of the number of new spawned fish in a population in the Pacific
Ocean. Since the fish are known to spawn better in colder waters~\citep{rosen2007automatic}, it
can be expected that there exists a cross-correlation between SOI and Recruitment.
The rescaled and centered version of the data that we analyze is shown in Figure~\ref{fig:soi_data}~(a).
\begin{figure}
  \centering
  \begin{subfigure}[c]{\textwidth} 
    \centering
    \includegraphics[width=\textwidth]{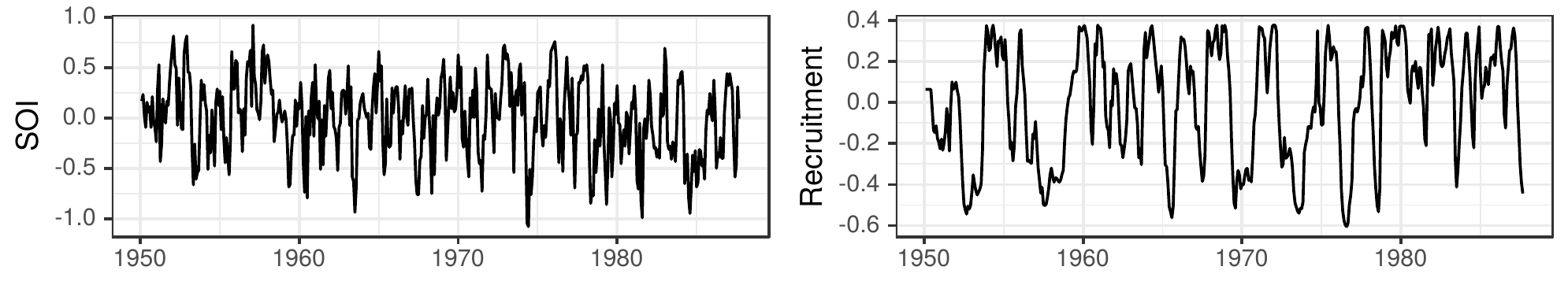}
  \subcaption{}
  \end{subfigure}
  \begin{subfigure}[c]{\textwidth} 
    \centering
    \includegraphics[width=\textwidth]{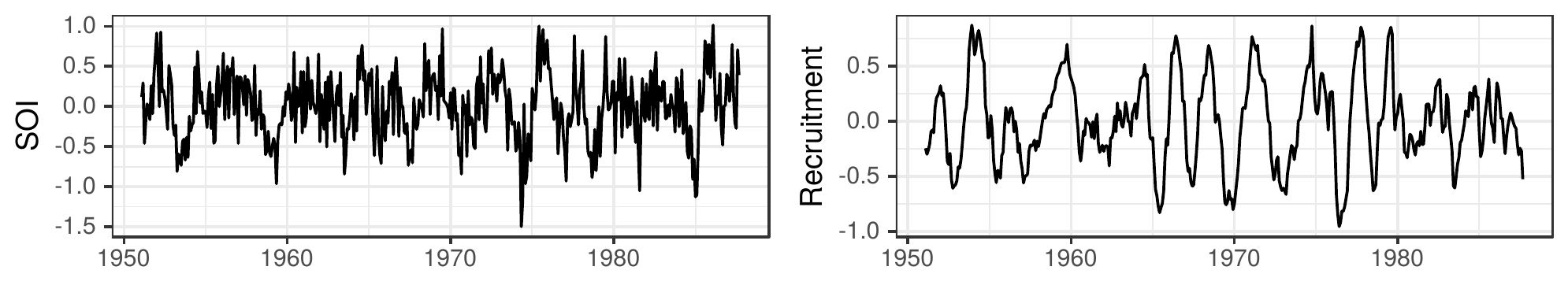}
  \subcaption{}
  \end{subfigure}
  \caption{Southern Oscillation Index and Recruitment time series in~(a) original version
  and~(b) lag 12 differenced version.}
  \label{fig:soi_data}
\end{figure}

The spectral inference results of the NP procedure are visualized in Figure~\ref{fig:soi_mWhittle}~(a).
\begin{figure}
  \centering
  \begin{subfigure}[c]{\textwidth} 
    \centering
    \includegraphics[width=.95\textwidth]{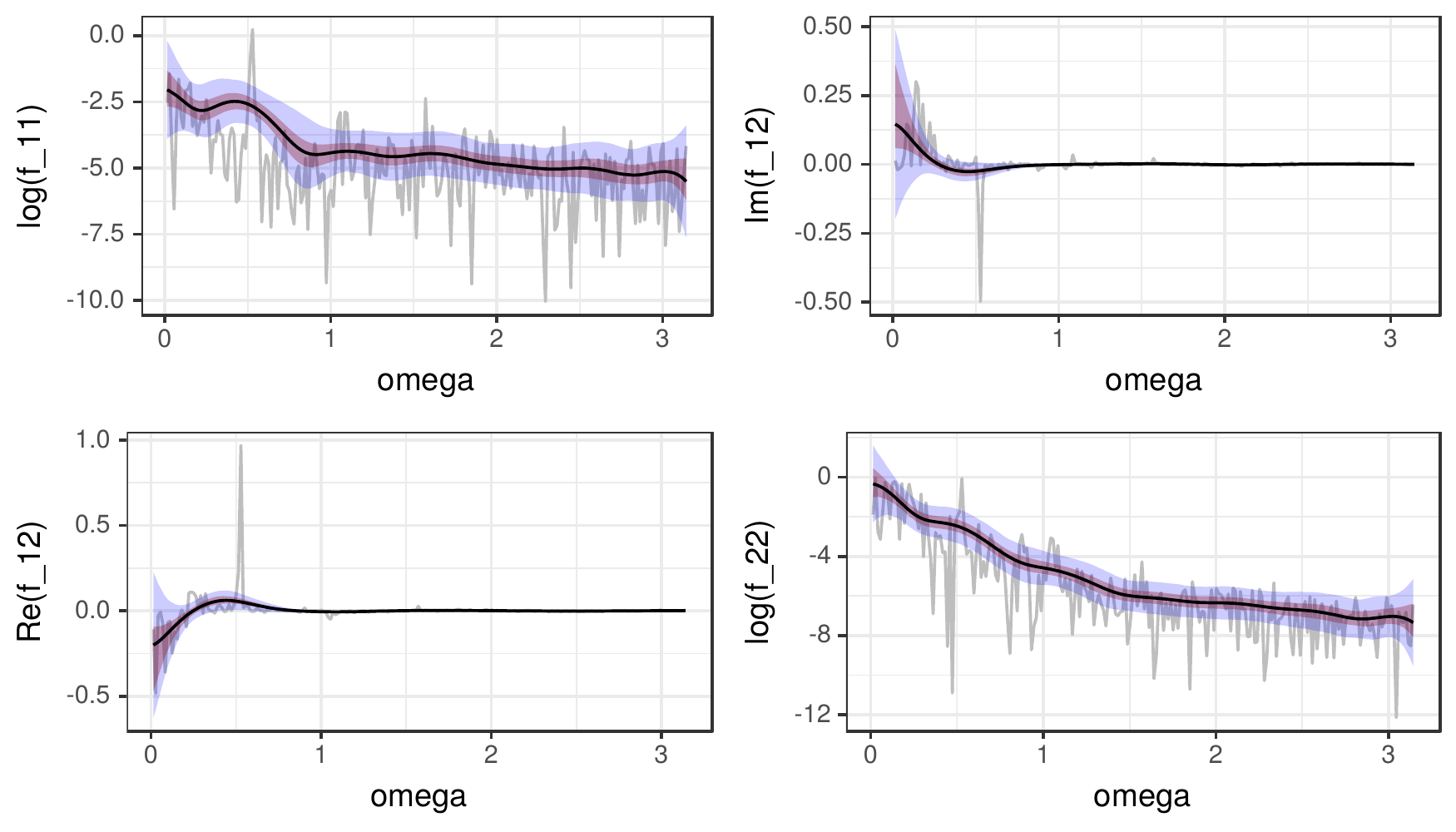}
  \subcaption{}
  \end{subfigure}
  \begin{subfigure}[c]{\textwidth} 
    \centering
    \includegraphics[width=.95\textwidth]{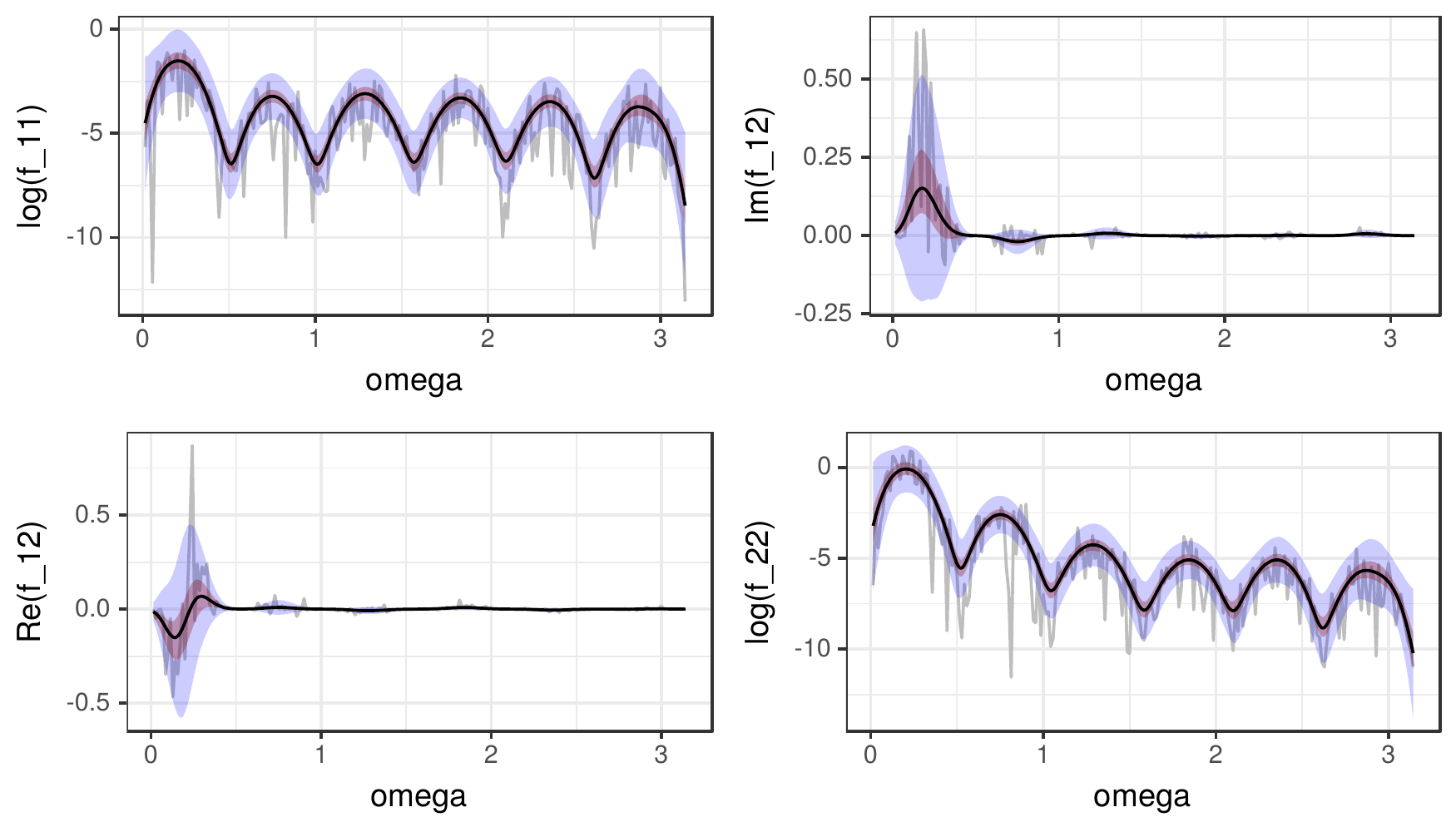}
  \subcaption{}
  \end{subfigure}
  \caption{Estimated spectra for the SOI and Recruitment series from the NP procedure
  for~(a) original data and~(b) data differenced at lag~12.
The posterior median spectral density is shown as solid black line, 
the pointwise 90\% region is visualized in shaded red
and the uniform 90\% region in shaded blue, whereas the periodogram is shown in gray.}
  \label{fig:soi_mWhittle}
\end{figure}
It can be seen that a spectral peak at the frequency~$\omega_{\operatorname{yearly}} := 2π/12 \approx 0.52$
corresponding to a temporal distance of 12~months is estimated in the individual spectrum ($f_{11}$) of the SOI time series. 
This is in line with the previous findings from~\cite{rosen2007automatic}.
We would like to emphasize that the pointwise and uniform posterior credibility regions
have to be interpreted with care, since we saw in the case of simulated data that
they are in general far from being honest frequentist confidence sets.

\par
To get a deeper insight into the dependency structure beyond the annual peak,
we also investigate the data differenced at lag~12, i.e.~the
time series~$\vect Y_t := \vect Z_t - \vect Z_{t-12}$.
The transformed data is visualized in Figure~\ref{fig:soi_data}~(b)
and the corresponding results of the NP procedure in Figure~\ref{fig:soi_mWhittle}~(b).
It can be seen that the annual peak at~$\omega_{\operatorname{yearly}}$ is not
present anymore.
However, since the lag 12 differenced model does not constitute a perfect fit,
the harmonics of the difference operator are visible in terms of six wavy bumps, particularly prominent
in the individual spectra.
On the other hand, the peak in the co-spectrum in the region of~$\omega\approx\omega_{\operatorname{yearly}}/4$
which is already visible in the original data, is even more prominent in the  differenced version.
This effect has also been observed by the authors in~\cite{rosen2007automatic}, 
who explained it by a possible El Ni\~{n}o effect.

\section{Conclusion}
In this work, we presented a new nonparametric prior for the spectral density matrix of a stationary multivariate time series.
Based on a mutual contiguity result, we established~$\mathbb L^1$-consistency and Hellinger contraction rates for Gaussian time series under the Hpd Gamma process based prior and Whittle's likelihood.
Regarding future work, it will be of interest to gain a better understanding of frequentist coverage properties of Bayesian credible sets, in particular if honest uncertainty quantification can be achieved, i.e.~credible regions that asymptotically match confidence regions.
Furthermore, it will be interesting to see if the asymptotic properties (at least consistency) carry over to different time series models beyond Gaussianity.
The main difficulty in this setting will be the lack of a mutual contiguity result.
Even for Gaussian data, it may be investigated whether the prior truncation of the eigenvalues can be dropped or at least relaxed towards bounds that are asymptotically growing to infinity (and shrinking to~0 respectively).

\par
In this work, we focused our attention on spectral density inference.
However, the Hpd Gamma process based approach can be applied to other models as well.
As an example, one may consider semiparametric models, where in a parametric model of interest 
(e.g.~linear regression, change point model), the time series constitutes the noise parameter,
which can be modeled nonparametrically.
From a theoretical perspective, it will be of interest whether a
Bernstein-von-Mises result can be established for the parameter of interest.
First considerations in this direction have been outlined in Section~9 in~\cite{meier}.

\subsection*{Acknowledgements}
The foundation of this work was laid when Alexander Meier was visiting the University of Auckland.
The visit has been supported by a travel scholarship (\emph{Kurzstipendien f\"ur Doktoranden}) of the German Academic Exchange Service (DAAD).
This work was supported by DFG grant AZ KI 1443/3-1.
Renate Meyer was also supported by the James Cook Research Fellowships from Government funding, administered by the Royal Society Te Ap\={a}rangi.

\appendix
\section{Proofs}\label{sec:app:proofs}

We will make extensive use of matrix algebra.
For a comprehensive overview on the most important
calculation rules and results, the reader may confer Appendix~B.1 in~\cite{meier}.
\subsection{Proofs of Section~\ref{sec:process}}
Detailed versions of the proofs in this section can be found in~\cite{meier},
Sections~3.1 and~3.2.

\begin{proof}[Proof of Theorem~\ref{th:distOfPhi}]
Consider the map~$\phi \colon \mathcal X \times \bar{\mathbb S}_d^+\times[0,\infty)\to[0,\infty)$ with $\phi(x,\matr U,r)=\|\indi_A(x)r\matr U)\|_T$.
By~\eqref{eq:nuAssumption}, the assumptions of Campbell's Theorem (see Section~3.2 in~\cite{kingman1992poisson}) are fulfilled
and an application thereof reveals that~$\sum_{(x,\matr U,r)\in\Pi}\phi(x,\matr U,r)$ is almost surely convergent.
Noting the representation~$\matr\Phi(A)=\sum_{(x,\matr U,r)\in\Pi}\indi_A(x)r\matr U$, this concludes
\[
  \|\matr\Phi(A)\|_T\leq\sum_{(x,\matr U,r)\in\Pi}\phi(x,\matr U,r)<\infty,
\]
hence~$\matr\Phi(A)\in\bar{\mathcal S}_d^+$ with probability one.
Let~$\matr\Theta\in\bar{\mathcal S}_d^+$ and consider the map~$\phi_{\matr\Theta}\colon\mathcal X \times \bar{\mathbb S}_d^+\times[0,\infty)\to[0,\infty)$ 
with~$\phi_{\matr\Theta}(x,\matr U,r)=\tr(\matr\Theta\indi_A(x)r\matr U)$.
A further application of Campbell's Theorem to~$\tr(\matr\Theta\matr\Phi(A))=\sum_{(x,\matr U,r)\in\Pi}\phi_{\matr\Theta}(x,\matr U,r)$ yields
\[
  \E\etr(-t\matr\Theta\matr\Phi(A))=
  \exp\left( - \int_A\int_{\bar{\mathbb S}_d^+}\int_0^\infty (1-\etr(-rt\matr\Theta\matr U)\nu(dx,d\matr U,dr) \right)
\]
for all~$t\geq 0$ (in particular for~$t=1$), concluding~(a).
To show~(b), first observe that with~$\nu_j(dx,d\matr U,dr):= \indi_{A_j}(x) \nu(dx,d\matr U,dr)$, it holds~$\nu = \sum_{j=1}^m \nu_j$.
Let~$\Pi_1,\ldots,\Pi_m$ be independent with~$\Pi_j\sim\operatorname{PP}(\nu_j)$.
By the Superposition Theorem (see Section~2.2 in~\cite{kingman1992poisson}), it follows 
that~$\Pi$ is equal in distribution to~$\cup_{j=1}^k \Pi_j$.
Since
\[
  \matr\Phi(A_j)
  \overset{d}{=} \sum_{(x,\matr U,r) \in \cup_{j=1}^k \Pi_j} \indi_{A_j}(x)r\matr U
  \overset{d}{=} \sum_{(x,\matr U,r) \in \Pi_j}r\matr U,
\]
the result follows. 
\end{proof}

The following result concerns full support of the Hpd Gamma process.
It will be used later in Lemma~\ref{lemma:priorPositivityFInfty} to show full support
of the Bernstein-Hpd-Gamma prior (as needed for posterior consistency).
For a probability measure~$\mu$, we denote by $\operatorname{supp}(\mu)$
the support of~$\mu$ and similarly, for a random variable~$X$, 
let~$\operatorname{supp}(X)$ denote the support of the distribution of~$X$.
\begin{lemma}\label{lemma:gpFullSupport}
Let~$\matr\Phi\sim\operatorname{GP}_{d\times d}(\alpha,\beta)$ with~$\alpha,\beta$
fulfilling Assumption~\ref{ass:gp1}.
Then for any measurable~$A \subset [0,\pi]$ with~$\int_Adx>0$ it holds that~$\operatorname{supp}(\matr\Phi(A))=\bar{\mathcal S}_d^+$.
\begin{proof}
Let~$\matr X_0 = r_0 \matr U_0 \in \mathcal S_d^+$.
We will show~$\matr X_0\in\operatorname{supp}(\matr X)$
and employ the usual decomposition of~$\matr X$ into a small jumps and large jumps component.
To elaborate, for~$\varepsilon \in (0,r_0)$
and~$\nu_A$ from~\eqref{eq:nuADef} we split~$\nu_A=\nu_{\matr X}+\nu_{\matr Y}$ with
\begin{equation}\label{eq:nuYZDef}\begin{split}
  \nu_{\matr Y}(d\matr U,dr)&:=\indi_{(\varepsilon/2,\infty)}(r)\nu_A(d\matr U,dr), \\
  \nu_{\matr Z}(d\matr U,dr)&:=\indi_{(0,\varepsilon/2]}(r)\nu_A(d\matr U,dr).
\end{split}\end{equation}
Let~$\matr Y,\matr Z$ be independent with \lvy measures~$\nu_{\matr Y}$ and~$\nu_{\matr Z}$
such that~$\matr\Phi(A)\overset{d}{=}\matr Y + \matr Z$.
Since~$\varepsilon < r_0$, it follows~$\matr X_0 \in \operatorname{supp}(\nu_{\matr Y})$
and hence, since~$\matr Y$ is compound Poisson,~$\matr X_0 \in \operatorname{supp}(\matr Y)$.
By Lemma 24.1 in~\cite{sato1999levy}, it suffices to show~$\matr 0 \in \operatorname{supp}(\matr Z)$,
which is equivalent to~$0 \in \operatorname{supp}(\tr \matr Z)$.
The \lvy measure on~$[0,\infty)$ of~$\tr \matr Z$ is given 
by
\begin{equation}\label{eq:nuTildeDef}
  \tilde\nu_{\matr Z}(dr):=\indi_{[0,\varepsilon/2]}(r)\nu(A,\bar{\mathbb S}_d^+,dr)
\end{equation}
for~$\nu$ from~\eqref{eq:meanMeasureProcess}.
Since it clearly holds~$\tilde\nu_{\matr Z}([0,\infty))=\infty$
and~$\int_0^1r\tilde\nu_{\matr Z}(dr)<\infty$, the assumptions of
Corollary~24.8 in~\cite{sato1999levy} are fulfilled, 
yielding~$0 \in \operatorname{supp}(\tr \matr Z)$.
\end{proof}
\end{lemma}

The results from Lemma~\ref{lemma:gpFullSupport} can be refined,
quantifying the probability mass that increments
of the Hpd Gamma process puts in small neighborhoods under
slightly stronger assumptions.
This is done in the following Lemma~\ref{lemma:gpMass1}.
It will be used later in Lemma~\ref{lemma:gpMass2} and Lemma~\ref{lemma:klMass} 
to derive lower probability bounds
for the Bernstein-Hpd-Gamma prior (as needed for posterior contraction rates).
\begin{lemma}\label{lemma:gpMass1}
Let~$\matr\Phi\sim\operatorname{GP}_{d\times d}(\alpha,\beta)$ fulfill Assumption~\ref{ass:gp2}.
Let~$\matr X_0\in\mathcal S_d^+$ with~$\|\matr X_0\|_T\leq\tau$ for some~$\tau > 1$.
Then for any~$A\subset[0,\pi]$ with~$\mathcal L(A)=\int_Adx>0$ with~$C_\alpha(A):=\int_A\alpha(x,\bar{\mathbb S}_d^+)dx$
it holds
\[
  P( \|\matr \Phi(A) - \matr X_0\| < \varepsilon )
  \geq C \kappa_\alpha(A)\exp((d^2+C_\alpha(A)+1)\log\varepsilon)
\]
for all~$\varepsilon>0$ small enough, where~$\kappa_\alpha(A)=\exp(-cC_\alpha(A))\mathcal L(A)$
and~$c,C>0$ are constants not depending on~$A$.
\begin{proof}
Fix~$A\subset [0,\pi]$ and let~$\varepsilon>0$ be small enough.
Denote by~$B_\varepsilon(\matr X_0)$ the ball with respect to~$\|\cdot\|_T$ in~$\mathcal S_d^+$.
Due to the equivalence of matrix norms, it suffices to provide the lower bound for~$P(\matr\Phi(A) \in B_\varepsilon(\matr X_0))$.
Let~$\matr Y,\matr Z$ be independent with \lvy measures~$\nu_{\matr Y},\nu_{\matr Z}$ from~\eqref{eq:nuTildeDef}
such that~$\matr\Phi(A)\overset{d}{=} \matr Y + \matr Z$.
Because of~$P(\matr\Phi(A) \in B_\varepsilon(\matr X_0)) \geq P(\matr Y \in B_{\varepsilon/2}(\matr X_0))P(\matr Z \in B_{\varepsilon/2}(\matr 0))$, it suffices to show the following assertions, for (generic)
constants~$c,C>0$ not depending on~$A$:
\begin{align}
\label{eq:gpMass1:toShow1}
&P(\matr Y \in B_{\varepsilon/2}(\matr X_0)) \geq C\kappa_\alpha(A)\exp((d^2+C_\alpha(A)+1)\log\varepsilon), \\
\label{eq:gpMass1:toShow2}
&P(\matr Z \in B_{\varepsilon/2}(\matr 0)) \geq C.
\end{align}

\par
Let~$\delta:=\varepsilon/2$.
Since~$\matr Y$ is Compound Poisson, we obtain from Theorem~4.3 in~\cite{sato1999levy} that
\[
  P(\matr Y\in B_{\delta}(\matr X_0))
  =\exp(-C_{\matr Y})\sum_{k\geq 0}\frac{1}{k!}\nu_{\matr Y}^k(B_{\delta}(\matr X_0))
  \geq\exp(-C_{\matr Y})\nu_{\matr Y}(B_{\delta}(\matr X_0))
\]
with~$C_{\matr Y}=\nu_{\matr Y}(\bar{\mathcal S}_d^+)\leq C_\alpha(A)E_1(\beta_0\delta)$
and~$E_1(x)=\int_x^\infty\exp(-r)/rdr$ denoting the exponential integral function.
Using~$E_1(x)\leq\log(1+1/x)$ (see~5.1.20 in~\cite{abramowitz1964handbook}), 
this leads to~$C_{\matr Y}\leq C_\alpha(A)(\log(3/\beta_0)-\log\varepsilon)$ and hence
\begin{equation}\label{eq:gpMass1:zwischenschritt}
  P(\matr Y\in B_{\delta}(\matr X_0))
  \geq \exp(-cC_\alpha(A)) \exp(C_\alpha(A)\log\varepsilon))\nu_A(B_{\delta}(\matr X_0)),
\end{equation}
with~$\nu_A$ from~\eqref{eq:nuADef},
where~$\nu_{\matr Y}(B_{\delta}(\matr X_0))=\nu_A(B_{\delta}(\matr X_0))$
was used.
Write~$\matr X_0=r_0\matr U_0$ with~$\matr U_0\in\mathbb S_d^+$ and~$r_0\in(0,\tau]$.
From~$\|\matr X-\matr X_0\|_T\leq\tau\|\matr U-\matr U_0\|+|r-r_0|$
we conclude that~$B_{\delta}(\matr X_0)$ is a superset of~$[-r_0-\delta/2,r_0+\delta/2]\times\tilde B(\matr U_0)$,
with~$\tilde B(\matr U_0)$ denoting the ball of radius~$\delta/(2\tau)$ in~$\mathbb S_d^+$.
This shows that~$\nu_A(B_{\delta}(\matr X_0))$ is bounded from below by
\begin{align*}
  &\int_A\int_{\tilde B(\matr U_0)}\int_{r_0-\delta/2}^{r_0+\delta/2}\frac{\exp(-\beta(x,\matr U)r)}{r}dr\alpha(x,d\matr U)dx\\
  &\geq \alpha_A(\tilde B(\matr U_0))\int_{r_0-\delta/2}^{r_0+\delta/2}\frac{\exp(-\beta_1r)}{r}dr
\end{align*}
with the measure~$\alpha_A(d\matr U):=\int_A\alpha(x,d\matr U)dx$.
Using
\[
  \int_{r_0-\delta/2}^{r_0+\delta/2}\frac{\exp(-\beta_1r)}{r}dr
  \geq \frac{\delta \exp(-\beta_1(r_0+\delta/2))}{r_0+\delta/2}
  \geq C\varepsilon
\]
and (conceiving~$\mathbb S_d^+$ as a subset of~$\mathbb R^{d^2}$)
\[
  \alpha_A(\tilde B(\matr U_0)) 
  \geq g_0\mathcal L(A)\int_{\tilde B(\matr U_0)}d\matr U
  \geq C\mathcal L(A)\varepsilon^{d^2}
\]
leads to~$\nu_A(B_{\delta}(\matr X_0))\geq C\varepsilon^{d^2+1}\mathcal L(A)$.
Together with~\eqref{eq:gpMass1:zwischenschritt}, this yields~\eqref{eq:gpMass1:toShow1}.

\par
To show~\eqref{eq:gpMass1:toShow2}, we will apply Lemma~26.4 in~\cite{sato1999levy}.
Recall the \lvy measure~$\tilde\nu_{\matr Z}$
of~$\tr\matr Z$ from~\eqref{eq:nuTildeDef}
and consider the function~$\psi(u):=\int_{[0,1]} ( \exp(ur)-1-ur )\tilde\nu_{\matr Z}(dr)$
for~$u \in \mathbb R$.
Denote by~$u=u(\xi)$ the inverse function of~$\xi=\psi'(u)$ for~$\xi\in(0,\infty)$.
Using the Lambert~$W$ function (see~\cite{corless1996lambertw}),
one can show that
\[
  u(\xi) 
  \geq\frac{\xi}{\delta^2}W\left(\frac{2}{C_\alpha}\right), \quad \xi\in (0,\delta),
\]
with~$C_\alpha:=C_\alpha([0,\pi])<\infty$.
By an application of Lemma~26.4 in~\cite{sato1999levy} we get
\[
  P(\tr\matr Z \geq \delta)
  \leq\exp\left( -\int_0^\delta u(\xi)d\xi \right)
  \leq\exp\left(-\frac{1}{2}W(2/C_\alpha) \right)<1,
\]
concluding~\eqref{eq:gpMass1:toShow2}.
\end{proof}
\end{lemma}

\subsection{Proofs of Section~\ref{sec:contiguity}}\label{sec:app:proof:contiguity}
For a detailed proof of all the results in this section, the reader is referred
to Section~4.2 and Section~4.3 in~\cite{meier}.
To establish Theorem~\ref{th:contiguity} with methods from probability theory,
we translate the complex-valued version
of Whittle's likelihood  into a real-valued version, which is
obtained by representing the real and imaginary parts
separately.
Indeed, it can be shown that~$\tilde P_W^n$ can equivalently be written as
\begin{equation}\label{eq:whittleRealValued}
  p_W^n(\vectt z | \matr f) = \frac{1}{\sqrt{(2\pi)^{nd} |\matr D_{nd}|}}
  \exp\left(-\frac{1}{2}\vectt z^T \matr D_{nd}^{-1} \vectt z \right), \quad \vectt z \in \mathbb R^{nd}
\end{equation}
with
\begin{equation}\label{eq:dndMatDef}
  \matr D_{nd}=\matr D_{nd}[\matr f]= \operatorname{diag}( \matr f(0),\mathcal B\matr f(\omega_1),\ldots,\mathcal B\matr f(\omega_N) )
\end{equation}
being the block diagonal with the blocks $\matr f(0),\mathcal B\matr f(\omega_1),\ldots,\mathcal B\matr f(\omega_N)$ 
(and~$\matr f(\pi)$ for~$n$ even) and the algebra isomorphism
\[
  \mathcal B \colon \mathbb C^{d\times d}\to\mathbb R^{2d\times 2d}, \quad \mathcal B\matr A = \begin{pmatrix} \Re \matr A & -\Im\matr A \\ \Im\matr A & \Re\matr A \end{pmatrix},
\]
see~p.224 in~\cite{hannan}.
Similarly, the full Gaussian likelihood~$\tilde P^n$ can be written 
in the frequency domain in terms of~$\vectt Z$ as
\begin{equation}\label{eq:gaussianFrequency}
  \tilde p^n(\vectt z | \matr f) = \frac{1}{\sqrt{(2\pi)^{nd} |\matr \Gamma_{nd}|}}
  \exp\left(-\frac{1}{2}\vectt z^T (\matr F_{nd}\matr \Gamma_{nd}\matr F_{nd}^T)^{-1} \vectt z \right), \quad \vectt z \in \mathbb R^{nd},
\end{equation}
where~$\matr\Gamma_{nd}:=\matr\Gamma_{nd}[\matr f]:=(\matr\Gamma(-i+j))_{i,j=0}^{n-1}\in\mathbb R^{nd\times nd}$ 
is the time-domain covariance matrix
with~$\matr\Gamma(h)=\int_0^{2\pi}\matr f(\omega)\exp(ih\omega)d\omega\in\mathbb R^{d\times d}$
and~$\matr F_{nd}\in\mathbb R^{nd\times nd}$ is the orthogonal Fourier transformation matrix.
It can be shown that~$\matr F_{nd}=\matr R_{nd}^T(\matr I_d \otimes \matr F_n) \matr R_{nd}$,
where~$\matr R_{nd}$ is a permutation matrix fulfilling
\[
  \matr R_{nd}(X_1,X_2,\ldots,X_{nd})=(X_1,X_{d+1},\ldots,X_{(n-1)d+1},X_2,X_{d+2},\ldots,X_{nd})
\]
and~$\matr F_n$ is the univariate Discrete Fourier Transform (DFT) matrix (see Section~10.1 in~\cite{brockwell}).

We need the following auxiliary results.
Lemma~\ref{lemma:szegoe} is a classic result that links the determinant of the time-domain
Block Toeplitz covariance matrix with the integrated log spectral density.
Lemma~\ref{lemma:igammaGammInv} asymptotically bounds the norm of~$\matr I_{nd}-\matr\Gamma_{nd}\matr\Gamma_{nd}^-$
and Lemma~\ref{lemma:circulantApproximation}
establishes rates when approximating the Gaussian frequency domain covariance matrix
by a block circulant matrix (see e.g.~\cite{gray2006toeplitz} for an overview of this topic),
which is closely related to the covariance matrix under Whittle's likelihood.
\begin{lemma}[Szeg\"os Strong Limit Theorem for Block Toeplitz Matrices]\label{lemma:szegoe}
Let~$\matr f$ fulfill Assumption~\ref{ass:f0}.
Then, with~$G:=\exp(\frac{1}{2\pi}\int_0^{2\pi}\log|2\pi\matr f(\omega)|d\omega)>0$
it holds~$|\matr\Gamma_{nd}|/G^n\to E$ as~$n\to\infty$, where~$E$ is
a positive constant.
\begin{proof}
The result follows from Theorem~10.30 in~\cite{boettcher}.
\end{proof}
\end{lemma}

\begin{lemma}\label{lemma:igammaGammInv}
Let~$\matr f$ fulfill Assumption~\ref{ass:f0}.
Then~$\|\matr I_{nd} - \matr\Gamma_{nd}\matr\Gamma_{nd}^{-}\|=O(1)$ as~$n \to \infty$.
\begin{proof}
The proof for the case~$d=1$ can be found in Lemma A1.4 in~\cite{dz}
and the proof for the case~$d>1$ follows along the lines.
\end{proof}
\end{lemma}

We will employ the following notational convention:
For a matrix~$\matr A \in \mathbb R^{nd\times nd}$,
denote the (disjoint)~$d\times d$ blocks of~$\matr A$ 
by~$\matr A(i,j)\in\mathbb R^{d\times d}$ 
for~$i,j=1,\ldots,n$.
Furthermore, for two positive sequences~$(a_n),(b_n)$
we write~$a_n\lleq b_n$ if there exists a constant~$c>0$
such that~$a_n\leq cb_n$ for all~$n$.

\begin{lemma}\label{lemma:circulantApproximation}
Let~$\matr f$ fulfill Assumption~\ref{ass:f0} 
and~$\matr H_{nd}:=\matr F_{nd}\matr\Gamma_{nd}\matr F_{nd}^T - \matr D_{nd}$
with~$\matr D_{nd}$ from~\eqref{eq:dndMatDef}.
Then there exists a constant~$C>0$ such that~$\|\matr H_{nd}(i,j)\|\leq Cn^{-1}$
holds for~$i,j=1,\ldots,n$ and all~$n$.
\begin{proof}
Consider the symmetric block circulant matrix~$\matr \Gamma^{\circ}_{nd}\in\mathbb R^{nd\times nd}$
defined as
\begin{align*}
  \begin{pmatrix}
    \matr\Gamma(0)                               &  \matr\Gamma(1)    &  \ldots            
  & \matr\Gamma\left(\lfloor n/2 \rfloor \right) &  \matr\Gamma\left(\lceil n/2 \rceil -1 \right)^T & \ldots
  & \matr\Gamma(2)^T                             &  \matr\Gamma(1)^T \\ 
    \matr\Gamma(1)^T                             &  \matr\Gamma(0)    &  \ldots            
  & \matr\Gamma\left(\lfloor n/2 \rfloor -1 \right) &  \matr\Gamma\left(\lfloor n/2 \rfloor \right) & \ldots
  & \matr\Gamma(3)^T                             &  \matr\Gamma(2)^T \\ 
  \vdots & \vdots && \vdots & \vdots && \vdots & \vdots \\
    \matr\Gamma(1)                               &  \matr\Gamma(2)    &  \ldots            
  & \matr\Gamma\left(\lceil n/2 \rceil -1 \right)^T &  \matr\Gamma\left(\lceil n/2 \rceil -2 \right)^T & \ldots
  & \matr\Gamma(1)^T                             &  \matr\Gamma(0)^T \\ 
  \end{pmatrix}.
\end{align*}
Let~$\matr G_{nd}:=\matr F_{nd}(\matr\Gamma_{nd}-\matr\Gamma^{\circ}_{nd})\matr F_{nd}^T$.
Using the representation of the block components $\matr G_{nd}(i,j)=\sum_{k,l=1}^n \matr F_{nd}(i,k) ( \matr\Gamma_{nd}(k,l) - \matr\Gamma^{\circ}_{nd}(k,l)  ) \matr F_{nd}(j,l)^T$
as well as $\|\matr F_{nd}(i,j)\|\lleq n^{-1/2}$, we compute with~$N=\lceil n/2 \rceil - 1$
\begin{align*}
  n\| \matr G_{nd}(i,j) \| 
  &\lleq  \sum_{k,l=1}^n \| \matr\Gamma_{nd}(k,l) - \matr\Gamma^{\circ}_{nd}(k,l) \| \\
  &\lleq  \sum_{m=1}^{N}m\|\matr\Gamma(m)\| + N\sum_{l=n-N}^{n-1}\|\matr\Gamma(l)\|
  =O(1)
\end{align*}
by Assumption~\ref{ass:f0}.
This shows~$\|\matr G_{nd}(i,j)\| \lleq n^{-1}$ uniformly in~$i,j$ as~$n \to \infty$.
Now consider $\matr f_n(\omega)=\frac{1}{2\pi}\sum_{|h|\leq\lfloor n/2 \rfloor}\matr\Gamma(h)\exp(-ih\omega)$ 
and the corresponding 
block diagonal matrix~$\matr D_{nd}[\matr f_n]\in\mathbb R^{nd\times nd}$ of~$\matr f_n$ as in~\eqref{eq:dndMatDef}.
Using Proposition~4.5.1 in~\cite{brockwell}, a few elementary calculations 
show $\matr F_{nd}\matr\Gamma^{\circ}_{nd}\matr F_{nd}^T=\matr D_{nd}[\matr f_n]$,
and hence~$\matr H_{nd}(i,j)=\matr D_{nd}[\matr f_n](i,j) - \matr D_{nd}[\matr f](i,j) + O(n^{-1})$
uniformly in~$i,j$.
Since it also holds
\[
  n\|\matr f_n(\omega) - \matr f(\omega)\|
  \leq \frac{n}{2\pi}\sum_{|h|>\lfloor n/2 \rfloor} \|\matr\Gamma(h)\|
  \lleq \sum_{|h|>\lfloor n/2 \rfloor} |h|\|\matr\Gamma(h)\|
  = O(1)
\]
for~$0 \leq \omega \leq \pi$ by Assumption~\ref{ass:f0} and hence~$\| \matr D_{nd}[\matr f_n](i,j)- \matr D_{nd}[\matr f](i,j)\| \lleq n^{-1}$ uniformly in~$i,j$ as~$n \to \infty$, the claim follows.
\end{proof}
\end{lemma}

Now we can present the proof of the contiguity result.
\begin{proof}[Proof of Theorem~\ref{th:contiguity}]
It suffices to show that~$\tilde P_W^n$ and~$\tilde P^n$ from~\eqref{eq:whittleRealValued}
and~\eqref{eq:gaussianFrequency} are mutually contiguous.
To this end, following the arguments from the proof of the univariate case 
in~\cite{choudhuriContiguity}, it suffices to show
that the sequence of random variables
\[
  \Lambda_n 
  = \log \frac{\tilde p_W^n(\vectt Z)}{\tilde p^n(\vectt Z)}
  = \frac{1}{2}(\log |\matr\Gamma_{nd}| - \log |\matr D_{nd}|)
    + \frac{1}{2} \vectt Z^T \left( (\matr F_{nd}\matr \Gamma_{nd}\matr F_{nd}^T)^{-1}-\matr D_{nd}^{-1} \right) \vectt Z
\]
has uniformly bounded mean and variance under both~$\tilde P_W^n$ and~$\tilde P^n$.
Using the result from Lemma~\ref{lemma:szegoe}, the boundedness of~$\log |\matr\Gamma_{nd}| - \log |\matr D_{nd}|$
as~$n\to\infty$ follows with the same argument as for the univariate case, see~\cite{choudhuriContiguity}.
Letting~$\tilde\Lambda_n:=\vectt Z^T ( \tilde{\matr\Gamma}_{nd}^{-1}-\matr D_{nd}^{-1} ) \vectt Z$
with~$\tilde{\matr\Gamma}_{nd}=\matr F_{nd}\matr \Gamma_{nd}\matr F_{nd}^T$,
it remains to show that each of the following sequences is bounded:
\begin{align*}
  \E_{\tilde P^n}\tilde \Lambda_n = \tr\left(\matr I_{nd}-\tilde {\matr \Gamma}_{nd}\matr D_{nd}^{-1}\right), \quad
  \Var_{\tilde P^n}\tilde \Lambda_n = 2\tr\left(\left(\matr I_{nd}-\tilde {\matr \Gamma}_{nd}\matr D_{nd}^{-1}\right)^2\right), \\
  \E_{\tilde P_W^n}\tilde \Lambda_n = \tr\left( \matr D_{nd}\tilde {\matr \Gamma}_{nd}^{-1}-\matr I_{nd}\right), \quad
  \Var_{\tilde P_W^n}\tilde \Lambda_n = 2\tr\left(\left( \matr D_{nd}\tilde {\matr \Gamma}_{nd}^{-1}-\matr I_{nd}\right)^2\right).
\end{align*}
Let~$\matr H_{nd}$ be defined as in Lemma~\ref{lemma:circulantApproximation}.
Using~$|\tr(\matr A\matr B)| \leq \|\matr A\|_2 \|\matr B\|$ with~$\|\matr A\|_2$
denoting the largest singular value, we get
\[
  |\E_{\tilde P^n}\tilde \Lambda_n |
  = |\tr(\matr H_{nd}\matr D_{nd}^{-1})|
  \lleq \sum_{j=0}^{\lfloor n/2 \rfloor} \| \matr H_{nd}(j,j) \| \| \matr f(\omega_j)^{-1} \|_2
  = O(1)
\]
by Lemma~\ref{lemma:circulantApproximation} and Assumption~\ref{ass:f0}.
A similar calculation that follows the arguments in~\cite{choudhuriContiguity}
shows~$\Var_{\tilde P^n}\tilde \Lambda_n=O(1)$.
For the mean and variance under Whittle's likelihood, consider a time series 
with spectral density matrix~$\matr f^{-1}$.
From Lemma 13.3.2 in~\cite{grochenig2013foundations}, we get that~$\matr f^{-1}$
also fulfills Assumption~\ref{ass:f0}.
Using~$\| \matr \Gamma_{nd}^{-1} \|_2 \leq \max_{0\leq\omega\leq\pi}\|\matr f(\omega)^{-1}\|_2=O(1)$
from Lemma~2.1 in~\cite{hannanWahlberg}
and the result from Lemma~\ref{lemma:igammaGammInv},
the proof of~$\E_{\tilde P_W^n}\tilde \Lambda_n=O(1)$ and~$\Var_{\tilde P_W^n}\tilde \Lambda_n=O(1)$
follows along the lines of~\cite{choudhuriContiguity}.
\end{proof}

\subsection{Proofs of Section~\ref{sec:consistency}}\label{sec:app:proof:consistency}
In this Section we will present the proof of Theorem~\ref{th:consistency}.
It relies on the contiguity result from Corollary~\ref{cor:contiguity}.
The proof technique consists of two main parts: prior positivity of neighborhoods 
and existence of exponentially powerful tests,
both which are discussed in the following.
Detailed proofs of all results in this section can be found in~\cite{meier}, Section~7.1.
\begin{proof}[Proof of Theorem~\ref{th:consistency}]
We will apply a general consistency theorem 
for non-iid observations from~\cite{choudhuri}, see Theorem~A.1 in~\cite{choudhuri}.
The prior KL support 
and testability 
assumptions are verified in Lemma~\ref{lemma:KLpositivity} 
and Lemma~\ref{lemma:uniformTestability} below.
It thus remains to bound the prior mass of the sieve complement:~$P_{0,\tau}(\Theta_n^c)\lleq \sum_{k>k_n}p(k)$,
which is bounded from above by~$\exp(-cn)$ for a constant~$c>0$
by Assumption~\ref{ass:gp1} concluding the proof.
\end{proof}

\subsubsection{Prior positivity of neighborhoods}
We start our considerations with prior posivity of uniform neighborhoods.
The maximum Frobenius 
norm is defined as~$\|\matr f\|_{F,\infty}:=\max_{0\leq\omega\leq\pi}\|\matr f(\omega)\|$.
\begin{lemma}\label{lemma:priorPositivityFInfty}
Let the assumptions of Theorem~\ref{th:consistency} be fulfilled.
Then with~$B_\varepsilon:=\{ \matr f \colon \|\matr f-\matr f_0\|_{F,\infty}<\varepsilon \}$
it holds~$P_{0,\tau}( B_\varepsilon )>0$ for every~$\varepsilon>0$.
\begin{proof}
The proof is analogously to the proof of~(B.1) in~\cite{choudhuri},
using the insight from Lemma~\ref{lemma:gpFullSupport},
observing that for~$\varepsilon$ small enough it holds~$B_\varepsilon \subset \mathcal C_{0,\tau}$.
\end{proof}
\end{lemma}
The result from Lemma~\ref{lemma:priorPositivityFInfty} can be used to show prior positivity of Kullback Leibler (KL) neighborhoods,
as summarized in the following Lemma.
Recall the definition of the Fourier coefficients~$\vectt Z_1,\ldots,\vectt Z_N$ with~$N=\lceil n/2 \rceil-1$ from~\eqref{eq:fourierCoef}
and their respective probability densities~$p_{j,N}$ from~\eqref{eq:pjn} under Whittle's likelihood~$P_W^n$.
Consider the KL terms~$K_N(\matr f_0,\matr f):=\frac{1}{N}\sum_{j=1}^NK_{j,N}(\matr f_0,\matr f)$
and~$V_N(\matr f_0,\matr f):=\frac{1}{N}\sum_{j=1}^NV_{j,N}(\matr f_0,\matr f)$ with
\begin{equation}\label{eq:klDef}\begin{split}
  K_{j,N}(\matr f_0,\matr f) &= \E_{\matr f_0}\log\frac{p_{j,n}(\vectt Z_j | \matr f_0)}{p_{j,n}(\vectt Z_j | \matr f_0)},
  \\
  V_{j,N}(\matr f_0,\matr f) &= \Var_{\matr f_0}\log\frac{p_{j,n}(\vectt Z_j | \matr f_0)}{p_{j,n}(\vectt Z_j | \matr f_0)},
\end{split}\end{equation}
where~$E_{\matr f_0}$ and~$\Var_{\matr f_0}$ denote mean and variance under~$P_W^n(\cdot | \matr f_0)$.
\begin{lemma}\label{lemma:KLpositivity}
Under the assumptions of Theorem~\ref{th:consistency}, it holds
\begin{align}
  \label{eq:KLpositivity1}
  &\liminf_{N\to\infty}P_{0,\tau}\left( \matr f \in B_\varepsilon \colon K_N(\matr f_0,\matr f) < \frac{4\varepsilon^2}{b_0^2} \right) > 0, \\
  \label{eq:KLpositivity2}
  &\frac{1}{N}V_N(\matr f_0,\matr f) \to 0, \quad \text{for all } \matr f \in B_\varepsilon.
\end{align}
\begin{proof}
Using~$\lambda_1(\matr A + \matr B) \geq \lambda_1(\matr A) - \|\matr B\|$
for~$\matr A\in \mathcal S_d^+$ and~$\matr B \in\mathcal S_d$
we get~$\lambda_1(\matr f(\omega)) \geq \lambda_1(\matr f_0(\omega)) - \|\matr f - \matr f_0\|$,
which is larger or equal to~$b_0/2$ for all~$\matr f \in B_\varepsilon$ and all~$\varepsilon$ small enough.
Consider~$\matr Q(\omega):=\matr f(\omega)^{-1/2}\matr f_0(\omega)\matr f(\omega)^{-1/2}$,
with~$\matr A^{1/2}$ denoting the (unique) Hpd matrix square root of~$\matr A \in \mathcal S_d^+$.
Using~$\|\matr A \matr B\|\leq \lambda_1(\matr A^{-1})^{-1}\|\matr B\|$ for~$\matr A,\matr B\in\mathcal S_d^+$ 
yields
\[
  \sum_{i=1}^d\big(\lambda_i(\matr Q(\omega))-1\big)^2
  = \|\matr Q(\omega)-\matr I_d\|^2
  \leq \lambda_1(\matr f(\omega))^{-2}\|\matr f(\omega)-\matr f_0(\omega)\|
  \leq \frac{4\varepsilon^2}{b_0^2}\leq\frac{1}{4}.
\]
These considerations lead to
\begin{equation}\label{eq:KLpositivityTmp}
  \lambda_1(\matr f(\omega)) \geq \frac{b_0}{2}, \quad
  \lambda_1(\matr Q(\omega)) \geq \frac{1}{2}, \quad 0 \leq \omega \leq \pi.
\end{equation}
From
\[
  \log\frac{p_{j,n}(\vectt Z_j | \matr f_0)}{p_{j,n}(\vectt Z_j | \matr f_0)}
  = \log\frac{ |\matr f(\omega_j)|}{|\matr f_0(\omega_j)|}
  + \frac{1}{2\pi}\vectt Z_j^*\left( \matr f(\omega_j)^{-1} - \matr f_0(\omega_j)^{-1} \right)\vectt Z_j
\]
we arrive at
\begin{align*}
  K_{j,N}(\matr f_0,\matr f)
  &= \tr(\matr Q(\omega_j)-\matr I_d)-\log|\matr Q(\omega_j)|\\
  &= \sum_{i=1}^d\lambda_i(\matr Q(\omega_j))-1-\log\lambda_i(\matr Q(\omega_j)) \\
  & \leq \sum_{i=1}^d \big( \lambda_i(\matr Q(\omega_j))-1 \big)^2
  = \|\matr Q(\omega_j)-\matr I_d\|^2,
\end{align*}
where the inequality~$x-1-\log(x) \leq (x-1)^2$ for~$x \geq 1/2$ was used, 
recalling $\lambda_i(\matr Q(\omega)) \geq 1/2$ from~\eqref{eq:KLpositivityTmp}.
Using the finding~$\|\matr Q(\omega)-\matr I_d\|^2\leq 4\varepsilon^2/b_0^2$ from above,
this yields~$K_N(\matr f_0,\matr f)\leq 4\varepsilon^2/b_0^2$ for all~$\matr f \in B_\varepsilon$,
hence the probabilities on the left hand side of~\eqref{eq:KLpositivity1}
are bounded from below by~$P_{0,\tau}(B_\varepsilon)$, which is positive 
by Lemma~\ref{lemma:priorPositivityFInfty} and not depending on~$N$, hence~\eqref{eq:KLpositivity1}
follows.
By similar calculations using~\eqref{eq:KLpositivityTmp}, we also arrive at~$V_{j,N}(\matr f_0,\matr f))\leq 4\varepsilon^2/b_0$
for all~$\matr f \in B_\varepsilon$, yielding~$\frac{1}{N}V_N\leq 4\varepsilon^2/(b_0N)\to 0$
as~$N\to\infty$, concluding~\eqref{eq:KLpositivity2}.
\end{proof}
\end{lemma}

\subsubsection{Existence of tests}

Let~$\mathcal S_d^{+k}= \{ (\matr W_1,\ldots,\matr W_k) \colon \matr W_1,\ldots,\matr W_k\in\mathcal S_d^+\}$.
For~$\vect{\matr W}\in\mathcal S_d^{+k}$, 
let~$\mathfrak B(k,\vect{\matr W}):=\sum_{j=1}^k\matr W_jb(\cdot /\pi|j,k-j+1)$.
Let us consider the following sieve~$(\Theta_n)$ of the parameter space~$\Theta=\mathcal C_{0,\tau}$:
\begin{equation}\label{eq:sievedD}
  \Theta_n:= \bigcup_{k=1}^{k_n}
  \left\{ \mathfrak B(k,\vect{\matr W}) \colon \vect{\matr W} \in \mathcal S_d^{+k} \matr \right\} \cap \mathcal C_{0,\tau},
  \quad
  k_n := \left\lfloor \frac{\delta n}{\log n}  \right\rfloor,
\end{equation}
where~$\delta>0$ will be specified later.
A few calculations using the equivalence of matrix norms reveal
\begin{equation}\label{label:sieveZwischenschritt}
  \|\mathfrak B(k,\vect{\matr W}) \|_{F,\infty} \geq d^{-1/2} \max_{j=1,\ldots,k}\|\matr W_j\|.
\end{equation}

The following result quantifies the metric entropy of~$\Theta_n$ 
in terms of~$\varepsilon$-covering numbers (see Appendix~C in~\cite{ghosal2017fundamentals}).
It is similar in spirit to Lemma~B.4 in~\cite{choudhuri}
and will be needed later for the construction of uniformly exponentially powerful
tests from tests against fixed alternatives, where the alternative set
is covered by small balls and the number of balls is controlled to ensure test consistency.
\begin{lemma}\label{lemma:coveringNumber}
The~$\varepsilon$-covering number of~$\Theta_n$ with respect to~$\|\cdot\|_{F,\infty}$
is bounded by
\[
  \log N(\varepsilon,\Theta_n,\|\cdot\|_{F,\infty}) 
  \lleq k_n\left(\log k_n + \log\frac{\tau}{\varepsilon}\right).
\]
\begin{proof}
For~$k \leq k_n$, let~$\vect{\matr W}_1,\vect{\matr W}_2\in\mathcal S_d^{+k}$
such that~$\matr f_i:=\mathfrak B(k,\vect{\matr W}_i)\in\Theta_n$ for~$i=1,2$.
With the norm~$\|\vect{\matr W_i}\|_1:=\sum_{j=1}^k\|\matr W_{ij}\|_1:=\sum_{j=1}^k\sum_{r,s=1}^d|W_{ijrs}|$,
we use~$\|\matr A\|_1 \leq d\|\matr A\|$ and~\eqref{label:sieveZwischenschritt}
to obtain~$\|\vect{\matr W}_i\|_1 \leq d\sum_{j=1}^k\|\matr W_{ij}\|\leq d^{3/2}k\tau$.
On the other hand, from~$|b(x|j,k-j+1)|\leq k$ and~$\|\matr A\|\leq\|\matr A\|_1$,
we also obtain~$\|\matr f_1-\matr f_2\|_{F,\infty}\leq k\|\vect{\matr W}_1-\vect{\matr W}_2\|_1$.
These considerations lead to
\begin{equation}\label{eq:epsilonCoveringZwischenschritt}
  N(\varepsilon,\Theta_n,\|\cdot\|_{F,\infty}) 
  \leq \sum_{k=1}^{k_n} N\left( \frac{\varepsilon}{k}, \left\{ \vect{\matr W}\in\mathcal S_d^{+k} \colon \|\vect{\matr W}\|_1 \leq d^{3/2}k\tau \right\},\|\cdot\|_1 \right).
\end{equation}
Conceiving~$\mathcal S_d^{+k}$ as a subset of~$\mathbb R^{2kd^2}$,
an application of~(A.9) in~\cite{ghosal2007convergence} shows that
each summand on the right hand side of~\eqref{eq:epsilonCoveringZwischenschritt}
is bounded from above by~$(6d^{3/2}k^2\tau/\varepsilon)^{2kd^2}\leq(6d^{3/2}k_n^2\tau/\varepsilon)^{2k_nd^2}$.
This leads to~$N(\varepsilon,\Theta_n,\|\cdot\|_{F,\infty}) \leq k_n(6d^{3/2}k_n^2\tau/\varepsilon)^{2k_nd^2}$
or
\begin{align*}
  \log N(\varepsilon,\Theta_n,\|\cdot\|_{F,\infty}) 
  &\leq k_n(4d^2\log k_n+2d^2\log(6d^{3/2}\tau/\varepsilon)+1)\\
  &\lleq k_n(\log k_n + \log(\tau/\varepsilon)).
\end{align*}
\end{proof}
\end{lemma}

The following result is related to Lemma~B.2 in~\cite{choudhuri}.
It translates the integral condition~$\int\|\matr f(\omega)-\matr f_0(\omega)\|d\omega>\varepsilon$
from the consistency complement sets to a suitable pointwise condition
that holds on sufficiently many Fourier frequencies.
This allows for testing at these frequencies.
Denote by~$\Omega_n:=\{ \omega_0,\ldots,\omega_N,\omega_{n/2} \}$ the set
of Fourier frequencies, with~$\omega_{n/2}$ being included if and only if~$n$ is even.
\begin{lemma}\label{lemma:numberFourierFreq}
Let~$\tau>0$ and~$\matr f_0 \in \mathcal C_{0,\tau}$. 
Then there exists~$k_0\in\mathbb N$ such that for every~$k\geq k_0$
and every polynomial~$\matr f\in\mathcal C_{0,\tau}$ of degree~$k$ with~$\int_0^\pi\|\matr f_0(\omega)-\matr f(\omega)\|d\omega>\varepsilon$
the function~$\matr Q(\omega):=\matr f_0(\omega)^{-1/2}\matr f(\omega)\matr f_0(\omega)^{-1/2}$
fulfills
\[
  \# \big\{ \omega\in\Omega_n\colon \lambda_1(\matr Q(\omega))<1-\tilde\varepsilon \text{ or } \lambda_d(\matr Q(\omega))>1+\tilde\varepsilon  \big\}
  \geq \frac{n\varepsilon}{8\pi\tau}-4k
\]
with~$\tilde\varepsilon=\varepsilon/(4\pi \sqrt d \tau)$.
\begin{proof}
Since~$\sum_{i=1}^d(\lambda_i(\matr Q(\omega))-1)^2>d\tilde\varepsilon^2$ implies
either~$\lambda_1(\matr Q(\omega))<1-\tilde\varepsilon$ or~$\lambda_d(\matr Q(\omega))>1+\tilde\varepsilon$,
we consider the set~$A:= \{ \omega \colon \sum_{i=1}^d(\lambda_i(\matr Q(\omega))-1)^2>d\tilde\varepsilon^2 \}$
and bound the cardinality of~$A\cap\Omega_n$ from below.
Using~$\|\matr A^{1/2}\matr B\matr A^{1/2}\|^2\geq\lambda_1(\matr A)\tr(\matr B\matr A\matr B)=\lambda_1(\matr A)\tr(\matr A\matr B\matr B)\geq\lambda_1(\matr A)^2\|\matr B\|^2$
for~$\matr A\in\mathcal S_d^+$ and~$\matr B \in \mathcal S_d$ 
(see~\cite{marshall1979inequalities}, p.341), it follows that
\[
  \sum_{i=1}^d\left( \lambda_i(\matr Q(\omega))-1 \right)^2
  = \|\matr Q(\omega)-\matr I_d\|^2
  = \| \matr f_0^{-1/2}(\omega) ( \matr f(\omega)-\matr f_0(\omega) ) \matr f_0^{-1/2}(\omega) \|^2
\]
is bounded from below by~$\tau^{-2}\|\matr f(\omega)-\matr f_0(\omega)\|$.
Hence the set~$B:=\{ \omega  \colon \|\matr f(\omega)-\matr f_0(\omega)\|>\delta\}$
with~$\delta:=\sqrt d\tau\tilde\varepsilon$
is a subset of~$A$ and we will continue to bound the cardinality of~$B\cap\Omega_n$ from below.
For~$k\geq k_0$, there exists a polynomial~$\matrt f_0\in\mathcal C_{0,\tau}$ of degree~$k$
with~$\max_{\omega}\|\matr f_0(\omega)-\matrt f_0(\omega)\|<\delta$.
It follows from~$\| \matr f(\omega) - \matrt f_0(\omega) \| \leq  \| \matr f(\omega) - \matr f_0(\omega) \| + \delta$
that the set~$C:=\{ \omega  \colon \|\matr f(\omega)-\matrt f_0(\omega)\|>2\delta\}$
is a subset of~$B$.
Denote by~$\mathcal L(C)$ the Lebesgue mass of~$C$.
Since
\begin{align*}
  \int_0^\pi \|\matr f(\omega)-\matrt f_0(\omega) \|d\omega
  &= \int_{C} \|\matr f(\omega)-\matrt f_0(\omega) \|d\omega
  + \int_{C^c} \|\matr f(\omega)-\matrt f_0(\omega) \|d\omega
\end{align*}
is bounded from above by~$2\tau\mathcal L(C)+\varepsilon/2$
and since~$\int\|\matr f_0-\matrt f_0\|d\omega \leq \varepsilon/4$, it follows
from~$\varepsilon<\int\|\matr f-\matr f_0\|d\omega\leq\int\|\matr f-\matrt f_0\|d\omega+\int\|\matr f_0-\matrt f_0\|d\omega$
that~$\mathcal L(C) \geq \varepsilon/(8\tau)$.
Observing~$C= \{\omega \colon t(\omega)>4\delta^2\}$ with~$t(\omega)=\tr((\matr f(\omega)-\matrt f_0(\omega))^2)$
being a polynomial of degree at most~$2k$, it follows with the same
argument as in the proof of Lemma~B2 in~\cite{choudhuri} that
the cardinality of~$C\cap\Omega_n$ is at least~$n\varepsilon/(8\pi\tau)-4k$.
\end{proof}
\end{lemma}

Based on these findings, we can now construct exponentially powerful tests.
We start with tests against fixed alternatives.
\begin{lemma}\label{lemma:fixedTestAlternative1}
Let~$\vect Y_1,\ldots,\vect Y_m$ be independent with~$\vect Y_j \sim CN_d(\vect 0,\matr \Sigma_j)$.
Consider testing~$H_0: \matr\Sigma_j=\matr\Sigma_{0j},j=1,\ldots,m$ 
against~$H_1: \matr\Sigma_j=\matr\Sigma_{0j},j=1,\ldots,m$ with
\begin{equation}\label{eq:fixedTestAlternative1}
  \lambda_d\left( \matr\Sigma_{0j}^{-1/2}\matr\Sigma_{1j}\matr\Sigma_{0j}^{-1/2} \right) > 1+\varepsilon, \quad j=1,\ldots,m.
\end{equation}
Then there exists a test~$\varphi_m$ and constants~$c_0,c_1>0$ only depending on~$\varepsilon$
such that $\E_{H_0}\varphi_m\leq\exp(-c_0m)$ and~$\E_{H_1}(1-\varphi_m)\leq\exp(-c_1m)$.
The result holds true analogously if
the condition~\eqref{eq:fixedTestAlternative1} is replaced 
by~$\lambda_1( \matr\Sigma_{0j}^{-1/2}\matr\Sigma_{1j}\matr\Sigma_{0j}^{-1/2} ) < 1-\varepsilon$
for~$j=1,\ldots,m$.
\begin{proof}
Let~$\matr Q_j:=\matr\Sigma_{0j}^{-1/2}\matr\Sigma_{1j}\matr\Sigma_{0j}^{-1/2}$
and~$\vect a_j:=\matr\Sigma_{0j}^{-1/2}\vect b_j$ with~$\vect b_j$
being an eigenvector of~$\matr Q_j$ corresponding to~$\lambda_d(\matr Q_j)$
for~$j=1,\ldots,m$.
Let~$\Psi_j:=\frac{\vect a_j^*\vect Y_j\vect Y_j^*\vect a_j}{\vect a_j^*\matr\Sigma_{0j}\vect a_j)}$
and consider the test statistic~$T_m:=\sum_{j=1}^m\Psi_j$.
We define the test as~$\varphi_m:=\indi_{\{T_m > m(1+\varepsilon/2)\}}$.
Under~$H_0$, it holds~$\Psi_j \iid \mathcal \chi_2^2/2$ (see the Appendix in~\cite{ibragimov1963estimation}),
hence~$T_m \overset{H_0}{\sim} \chi_{2m}^2/2$ and~$\E_{H_0}\varphi_m=P_{H_0}(T_m>m(1+\varepsilon/2))$.
A standard large deviation argument (as e.g.~in Lemma~A2 in~\cite{kirchMeyer})
concludes~$\E_{H_0}\varphi_m\leq\exp(-c_0m)$.
Under~$H_1$, it holds~$\Psi_j \overset{d}{=} \lambda_d(\matr Q_j)\vect X_j$
with~$\vect X_1,\ldots,\vect X_m \iid \chi_2^2/2$ by the Courant-Fisher Min-Max principle.
Letting $S_m:=\sum_{j=1}^m\vect X_j$, it follows from~\eqref{eq:fixedTestAlternative1} 
that
\[
  \E_{H_1}(1-\varphi_m) \leq P\left( \frac{S_m}{m} \leq \frac{1+\varepsilon/2}{1+\varepsilon} \right)
\]
and another large deviation argument concludes the proof.
\end{proof}
\end{lemma}

The previous result can now be used to construct uniform tests.
\begin{lemma}\label{lemma:uniformTestability}
Let the assumptions of Theorem~\ref{th:consistency} be fulfilled.
Let~$\varepsilon>0$ and~$U_\varepsilon:=\{\matr f \colon \int_0^\pi\|\matr f(\omega)-\matr f_0(\omega)\|d\omega<\varepsilon\}$.
With~$(\vectt Z_1,\ldots,\vectt Z_N)\sim P_W^n(\cdot|\matrt f)$, consider testing~$H_0: \matrt f =\matr f_0$
against~$H_1:\matrt f \in U_\varepsilon^c\cap\Theta_n$.
Then there exists a test~$\varphi_n$ fulfilling
\[
  \E_{\matr f_0}\varphi_n\to 0, \quad 
  \sup_{\matr f \in U_\varepsilon^c\cap\Theta_n}\E_{\matr f}(1-\varphi_n)\leq\exp(-cn)
\]
for a constant~$c>0$, with~$\E_{\matr f}$ denoting the mean under~$P_W^n(\cdot|\matr f)$.
\begin{proof}
Let~$\matr f \in U_\varepsilon^c\cap\Theta_n$.
We will show that there exists a test~$\varphi_{n,\matr f}$ fulfilling
\begin{equation}\label{eq:uniformTestZwischenschritt}
  \E_{\matr f_0} \varphi_{n,\matr f} \leq 2\exp(-c_0n),
  \quad
  \E_{\matr f} (1-\varphi_{n,\matr f}) \leq \exp(-c_1n)
\end{equation}
with positive constants~$c_0,c_1$ not depending on~$\matr f$.
The test~$\varphi_n$ as claimed
can then be constructed by the standard approach (see e.g.~Section~B.2 in~\cite{choudhuri} or Section~6.4 in~\cite{ghosal2017fundamentals})
of covering the alternative set~$U_\varepsilon^c\cap\Theta_n$ with small balls in conjunction with the
bound for the covering number from Lemma~\ref{lemma:coveringNumber}.

\par
We define the test~$\varphi_{n,\matr f}$ as follows.
Let~$\matr Q_j:=\matr f_0(\omega_j)^{-1/2}\matr f(\omega_j)\matr f_0(\omega_j)^{-1/2}$
and~$\vect a_j:=\matr f_0(\omega_j)^{-1/2}\vectt a_j$ with~$\vect 0 \neq \vectt a_j \in \mathbb C^d$ 
being an eigenvector of~$\matr Q_j$ corresponding
to~$\lambda_d(\matr Q_j)$
for~$j=1,\ldots,N$.
Similarly, let~$\vect b_j:=\matr f_0(\omega_j)^{-1/2}\vectt b_j$
with~$0 \neq \vectt b_j \in \mathbb C^d$ being an eigenvector of~$\matr Q_j$ corresponding
to~$\lambda_1(\matr Q_j)$.
With~$\tilde\varepsilon:=\varepsilon/(4\pi \sqrt d\tau)$,
consider
\begin{align*}
  \varphi_{n,\matr f}^+ &:= 1 \text{ if } \sum_{j \in I_{n,\matr f}^+} \frac{\vect a_j^*\vectt Z_j\vectt Z_j^*\vect a_j}{\vect a_j^*\matr f_0(\omega_j)\vect a_j}>m^+\left(1+\frac{\tilde\varepsilon}{2}\right),
  \quad \varphi_{n,\matr f}^+ := 0 \text{ else, } \\
  \varphi_{n,\matr f}^- &:= 1 \text{ if } \sum_{j \in I_{n,\matr f}^-} \frac{\vect b_j^*\vectt Z_j\vectt Z_j^*\vect b_j}{\vect b_j^*\matr f_0(\omega_j)\vect b_j}<m^-\left(1-\frac{\tilde\varepsilon}{2}\right),
  \quad \varphi_{n,\matr f}^- := 0 \text{ else,}
\end{align*}
with~$I_{n,\matr f}^+ := \left\{ j \colon \lambda_d(\matr Q_j)>1+\tilde\varepsilon \right\}$
and~$I_{n,\matr f}^- := \left\{ j \colon \lambda_1(\matr Q_j)<1-\tilde\varepsilon \right\}$
and~$m^{\pm}:=\# I_{n,\matr f}^\pm$.
By Lemma~\ref{lemma:numberFourierFreq} 
it holds~$m^++m^-\geq\frac{n\varepsilon}{8\pi\tau}-4k_n \geq \tilde\delta n$ 
with~$\tilde\delta>0$ if~$\delta$ from~\eqref{eq:sievedD} is chosen small enough.
Thus it holds either~$m^+\geq\tilde\delta n/2$ or~$m^-\geq\tilde\delta n/2$.
Since~$\E_{H_0} \varphi_{n,\matr f}^\pm \leq \exp(-c_0^\pm m^\pm)$
and~$\E_{H_1} (1-\varphi_{n,\matr f}^\pm) \leq \exp(-c_1^\pm m^\pm)$
by Lemma~\ref{lemma:fixedTestAlternative1},
it follows that~$\varphi_{n,\matr f} := \max \{ \indi_{\{ m^+ \geq \tilde \delta n \}}\varphi_{n,\matr f}^+, \indi_{\{ m^- \geq \tilde \delta n \}}\varphi_{n,\matr f}^- \}$ fulfills~\eqref{eq:uniformTestZwischenschritt}
with $c_0=\tilde\delta\min\{ c_0^+,c_0^- \}/2$ and~$c_1=\tilde\delta\max\{ c_1^+,c_1^- \}/2$.
\end{proof}
\end{lemma}

\subsection{Proofs of Section~\ref{sec:consistencyUniform}}\label{sec:app:proof:uniformConsisency}
\begin{proof}[Proof of Theorem~\ref{th:consietencyUniform}]
The proof is analogous to the proof of Theorem~\ref{th:consistency} in Section~\ref{sec:app:proof:consistency},
replacing Lemma \ref{lemma:priorPositivityFInfty} with Lemma \ref{lemma:fullSupportUniform} below
and Lemma \ref{lemma:numberFourierFreq} with Lemma \ref{lemma:numberOfFourierFrequenciesUniform} below.
\end{proof}

\begin{lemma}\label{lemma:fullSupportUniform}
Let the assumptions of Theorem~\ref{th:consietencyUniform} be fulfilled.
For~$\varepsilon>0$, let~$\tilde B_{\varepsilon} := \{ \matr f \colon \|\matr f-\matr f_0\|_L < \varepsilon \}$.
Then it holds~$\tilde P_\tau(\tilde B_\varepsilon)>0$.
\begin{proof}
It suffices to show~$P(\tilde B_\varepsilon)>0$ for~$\varepsilon>0$ small enough,
where~$P$ denotes the Bernstein-Hpd-Gamma prior (without restriction).
Consider the Bernstein polynomial approximation
\begin{equation}\label{eq:fkApprox}
  \matr f_{k}(\omega) := \sum_{j=1}^k \matr F_0(I_{j,k})b(\omega/\pi|j,k-j+1), \quad 0 \leq \omega \leq \pi
\end{equation}
of~$\matr f_0$ of degree~$k$.
Since~$\matr f_0$ is continuously differentiable by Assumption~\ref{ass:f0},
it follows that the derivative of~$\matr f_k$ converges
uniformly to the one of~$\matr f_0$ as~$k\to\infty$ (see Sections~1.4 and~1.8 in~\cite{lorentz2012bernstein}).
Hence, there exists~$k_0$ such that it holds~$\|\matr f_0-\matr f_k\|_L<\varepsilon/2$
for all~$k \geq k_0$.
This yields
\begin{align*}
  \| \matr f - \matr f_0 \|_L
  \leq \| \matr f - \matr f_k \|_L + \frac{\varepsilon}{2}
  \leq \sum_{j=1}^k \left\| \matr\Phi(I_{j,k})-\matr F_0(I_{j,k}) \right\| \|b_{j,k}\|_L + \frac{\varepsilon}{2}
\end{align*}
Using~$\| b_{j,k}' \|_L \leq 2k(k-1)+k \leq 3k^2$ (this follows from basic properties
of Bernstein polynomials, see Section~1.4 in~\cite{lorentz2012bernstein})
this leads to
\[
  \| \matr f - \matr f_0 \|_L
  \leq 3k^2 \sum_{j=1}^k \left\| \matr\Phi(I_{j,k})-\matr F_0(I_{j,k}) \right\| + \frac{\varepsilon}{2}
\]
and hence
\begin{equation}\label{eq:fullSupportUniformZwischenschritt1}
  P(\tilde B_\varepsilon)
  \geq \sum_{k\geq k_0}p(k)P\left( \max_{j=1,\ldots,k} \| \matr\Phi(I_{j,k})-\matr F_0(I_{j,k}) \| < \tilde\varepsilon_k \right)
\end{equation}
with~$\varepsilon_k=\varepsilon/(6k^3)$.
The right hand side of~\eqref{eq:fullSupportUniformZwischenschritt1} can be shown
to be positive with the same arguments as in the proof of~(B.1) in~\cite{choudhuri}.
\end{proof}
\end{lemma}

\begin{lemma}\label{lemma:numberOfFourierFrequenciesUniform}
Let~$\matr f_0\in\tilde{\mathcal C}_{b_1}$ be continuously differentiable.
Then there exists~$k_0 \in\mathbb N$ such that for every~$k \geq k_0$ and
every~$\matr f \in\tilde{\mathcal C}_\tau$ for~$\tau<b_1$
with components being polynomials of degree~$k$, it holds
true that~$\| \matr f-\matr f_0 \|_{F,\infty} > \varepsilon$ implies
\[
  \# \big\{ \omega\in\Omega_n\colon \lambda_1(\matr Q(\omega))<1-\tilde\varepsilon \text{ or } \lambda_d(\matr Q(\omega))>1+\tilde\varepsilon  \big\}
  \geq \frac{n\varepsilon}{16\pi\tau}-4k,
\]
with~$\tilde\varepsilon=\varepsilon/(8 \sqrt d \tau)$,
where~$\matr Q:=\matr f_0^{-1/2}\matr f\matr f_0^{-1/2}$
and~$\Omega_n$ denotes the set of Fourier frequencies.
\begin{proof}
Consider the set~$A:=\{x\in [0,\pi]\colon\|\matr f(x)-\matr f_0(x)\|>\varepsilon/8\}$.
Denote by~$\matrt f_0=\mathfrak B(k,\matr F_0)$ the Bernstein polynomial approximation
of degree~$k$ of~$\matr f_0$.
Since the set~$\tilde A:=\{x\in [0,\pi]\colon\|\matr f(x)-\matrt f_0(x)\|>\varepsilon/4\}$
fulfills~$\tilde A \subset A$, we continue by bounding the cardinality of~$\tilde A\cap\Omega_n$ from below.
For~$k$ large enough, it holds~$\| \matrt f_0-\matr f_0\|_L<\varepsilon/2$
and~$\| \matr f-\matrt f_0 \|_{F,\infty} >\varepsilon/2$.
Let~$x_0\in [0,\pi]$ with~$\| \matr f(x_0)-\matrt f_0(x_0) \|>\varepsilon/2$.
Then it follows for~$x_0+\xi\in[0,\pi]$ with the reverse triangle inequality
\begin{equation*}\begin{split}
  &\| \matr f(x_0+\xi)-\matrt f_0(x_0+\xi) \| \\
  &\geq \| \matr f(x_0)-\matrt f_0(x_0) \| - \| \matrt f_0(x_0)-\matrt f_0(x_0+\xi) \| - \| \matr f(x_0)-\matr f(x_0+\xi) \| \\
  &\geq \frac{\varepsilon}{2} - 2b_1|\xi|,
\end{split}\end{equation*}
yielding~$\|\matr f(x)-\matrt f_0(x)\|\geq\varepsilon/4$ for all~$x\in B := \{ x \in [0,\pi] \colon |x-x_0| \leq \frac{1}{8b_1}\varepsilon \}$.
The Lebesgue mass of~$B$ is at least~$\frac{1}{16b_1}\varepsilon$.
The same argument as in Lemma~\ref{lemma:numberFourierFreq} concludes the proof.
\end{proof}
\end{lemma}

\subsection{Proofs of Section~\ref{sec:contraction}}\label{sec:app:proof:rates}
The proof of Theorem~\ref{th:contraction} 
relies on a general contraction rate result from~\cite{ghosal2007convergence}
and is again split in two parts: Prior mass of neighborhoods (Lemma~\ref{lemma:klMass}) 
and and sieve entropy (Lemma~\ref{lemma:coveringNumberContr}).
For detailed proofs of all results in this section, see Section~7.2 in~\cite{meier}.
\begin{proof}[Proof of Theorem~\ref{th:contraction}]
Following the arguments of Section~9.5.2 in~\cite{ghosal2017fundamentals}, 
we use Lemma~\ref{lemma:klMass} below to find
\begin{equation}\label{eq:contractionProofLastRes}
  \frac{P_{\tau_0,\tau_1}(\Theta_n^c)}{P_{\tau_0,\tau_1}(B_{N,2}(\matr f_0,\varepsilon_n))}
  = o(\exp(-2n\varepsilon_n^2)).
\end{equation}
The result now follows by an application of Theorem~1 and Lemma~1 in~\cite{ghosal2007convergence},
noting that the assumptions thereof are fulfilled by Lemma~\ref{lemma:klMass} and
Lemma~\ref{lemma:coveringNumberContr} below, Lemma~2 in~\cite{ghosal2007convergence} 
and~\eqref{eq:contractionProofLastRes}. 
\end{proof}

\subsubsection*{Prior mass of neighborhoods}
Denote by~$\matr F_0(A):=\int_A\matr f_0(\omega)d\omega\in\bar{\mathcal S}_d^+$ 
for~$A\subset [0,\pi]$ the spectral measure corresponding to
spectral density~$\matr f_0$.
Recall that~$\{I_{j,k} \colon j=1,\ldots,k\}$ denotes the equidistant partition of~$[0,\pi]$
of size~$k$.
In the setting of Theorem~\ref{th:contraction}, we will consider the following sieve on~$\mathcal C_{\tau_0,\tau_1}$:
\begin{equation}\label{eq:sieveContraction}
  \Theta_n:= \bigcup_{k=1}^{k_n}
  \left\{ \mathfrak B(k,\vect{\matr W}) \colon \vect{\matr W} \in \mathcal S_d^{+k} \matr \right\} \cap \mathcal C_{\tau_0,\tau_1},
  \quad
  k_n := \left\lfloor \rho\varepsilon_n^{-2/a} \right\rfloor,
\end{equation}
where~$\rho>0$ will be specified later and~$\mathcal S_d^{+k}$ as defined before~\eqref{eq:sievedD}.
The following result quantifies the probability mass that~$\matr\Phi$
puts in neighborhoods of~$\matr F_0$.
\begin{lemma}\label{lemma:gpMass2}
Let~$\matr f_0$ fulfill Assumption~\ref{ass:f0}
and~$\matr\Phi\sim\operatorname{GP}_{d\times d}(\alpha,\beta)$
fulfill Assumption~\ref{ass:gp2}.
Then with~$\varepsilon_n$ as in Theorem~\ref{th:contraction}
and~$k_n$ from~\eqref{eq:sieveContraction} it holds
\begin{equation}\label{eq:gpMass2:claim}
  P\left( \sum_{j=1}^{k_n}\left\| \matr\Phi(I_{j,k_n})-\matr F_0(I_{j,k_n}) \right\|<\frac{\varepsilon_n}{k_n} \right)
  \geq C_0\exp(c_0k_n\log\varepsilon_n)
\end{equation}
for constants~$c_0,C_0>0$.
\begin{proof}
Using~$1/k_n\geq\varepsilon_n$ for~$n$ large enough, we find
\begin{align*}
  &P\left( \sum_{j=1}^{k_n}\left\| \matr\Phi(I_{j,k_n})-\matr F_0(I_{j,k_n}) \right\|<\frac{\varepsilon_n}{k_n} \right)\\
  &\geq P\left( \max_{j=1,\ldots,j_n} \left\| \matr\Phi(I_{j,k_n})-\matr F_0(I_{j,k_n}) \right\|< \varepsilon_n^3 \right) \\
  &= \prod_{j=1}^{k_n} P\left( \left\| \matr\Phi(I_{j,k_n})-\matr F_0(I_{j,k_n}) \right\|< \varepsilon_n^3 \right)
\end{align*}
where Theorem~\ref{th:distOfPhi}~(b) was used in the last step.
The assertion now follows from Lemma~\ref{lemma:gpMass1}, noting that with
\[
  \kappa_{\alpha}(I_{j,k_n})=\frac{\exp(-cC_{\alpha}(I_{j,k_n}))}{k_n}, \quad
   C_{\alpha}(I_{j,k_n})=\int_{I_{j,k_n}}\alpha(x,\bar{\mathbb S}_d^+)dx, 
\]
for~$j=1,\ldots,k_n$ it holds~$\sum_{j=1}^{k_n}C_{\alpha,j,k_n}=\int_0^\pi\alpha(x,\bar{\mathbb S}_d^+)dx<\infty$
as well as $\prod_{j=1}^{k_n}\kappa_{\alpha}(I_{j,k_n})\geq c/k_n^{k_n}\geq c\exp(k_n\log\varepsilon_n)$.
\end{proof}
\end{lemma}

The previous result can be used to quantify the probability mass
that the Bernstein-Hpd-Gamma prior puts in balls around~$\matr f_0$.
Recall the definition of the KL terms from~\eqref{eq:klDef} and
consider the KL neighborhoods
\[
  B_{N,2}(\matr f_0,\varepsilon):=\{ \matr f \in \Theta_n \colon K_N(\matr f_0,\matr f)<\varepsilon^2, V_N(\matr f_0,\matr f)<\varepsilon^2\}
\]
with~$\Theta_n$ from~\eqref{eq:sieveContraction}.
\begin{lemma}\label{lemma:klMass}
Let the assumptions of Theorem~\ref{th:contraction} be fulfilled.
Then there exists a constant~$c>0$ such that~$P_{\tau_0,\tau_1}(B_{N,2}(\matr f_0,\varepsilon_n)) \geq \exp(-cn\varepsilon_n^2)$.
\begin{proof}
The proof is analogous to 
Section~9.5.2 in~\cite{ghosal2017fundamentals}, using the result from Lemma~\ref{lemma:gpMass2}.
\end{proof}
\end{lemma}

\subsubsection*{Sieve entropy}
The next aim is to derive an upper bound for the~$\varepsilon$-covering
number of~$\Theta_n$ in the Hellinger topology.
We will use the bound in the maximum Frobenius topology from Lemma~\ref{lemma:coveringNumber}
for this purpose.
The following result established the needed link between the Hellinger and Frobenius topology.
\begin{lemma}\label{lemma:hellingerFrobeniusConnection}
Let~$\matr\Sigma_1,\matr\Sigma_2\in\mathcal S_d^+$ with~$\lambda_1(\matr\Sigma_i)\geq\tau_0$
and~$\lambda_d(\matr\Sigma_i)\leq\tau_1$ for~$i=1,2$.
Denoting by~$p_i$ the density of the~$CN_d(0,\matr\Sigma_i)$ distribution, 
it holds~$d_H^2(p_1,p_2) \lleq \|\matr\Sigma_1-\matr\Sigma_2\|^{1/2}$,
with proportionality constants only depending on~$\tau_0,\tau_1$ and~$d$.
\begin{proof}
A few computations (see the proof of Lemma~7.24 in~\cite{meier} for the details) yield
\[
  d_H^2(p_1,p_2)
  = 1-\int_{\vect z\in\mathbb C^d}\sqrt{p_1(\vect z)p_2(\vect z)}d\vect z
  = 1-\frac{|2\matr\Sigma_1^{1/2}\matr\Sigma_2^{-1/2}|}{|\matr\Sigma_1+\matr\Sigma_2|}
  = 1 - \frac{|2\matr I_d|}{|\matr Q+\matrt Q|}
\]
with~$\matr Q = \matr\Sigma_2^{-1/4}\matr\Sigma_1^{1/2}\matr\Sigma_2^{-1/4}\in\mathcal S_d^+$
and~$\matrt Q = \matr\Sigma_2^{-1/4}\matr\Sigma_1^{-1/4} \matr\Sigma_2 \matr\Sigma_1^{-1/4}\matr\Sigma_2^{-1/4}\in\mathcal S_d^+$
and~$\matr\Sigma^{1/4}$ denoting the Hpd matrix square root of~$\matr\Sigma^{1/2}$ for~$\matr\Sigma\in\mathcal S_d^+$.
Using~$|\matr Q+\matrt Q| \geq \prod_{i=1}^d(\lambda_i(\matr Q)+\lambda_i(\matrt Q)) \geq (2\tau_0)^d$
from~p.333 in~\cite{marshall1979inequalities}, we get
\[
  d_H^2(p_1,p_2)
  \lleq \left| \frac{\matr Q + \matrt Q}{2} \right| - |\matr I_d|
  \leq \lambda_d\left( \frac{\matr Q + \matrt Q}{2} \right)^d - 1
  \lleq \lambda_d\left( \frac{\matr Q + \matrt Q}{2} \right) - 1,
\]
where~$\lambda_d((\matr Q + \matrt Q)/2) \geq 1$ was used in the last step.
Using~$\lambda_d(\matr A)-1=\lambda_d(\matr A-\matr I_d)\leq\|\matr A-\matr I_d\|_2$, this leads to
\begin{equation}\label{eq:HFCZwischenschritt}
  d_H^2(p_1,p_2)\lleq
  \left\| \frac{\matr Q + \matrt Q}{2} -\matr I_d \right\|
  \lleq \|\matr Q - \matr I_d \| + \|\matrt Q - \matr I_d \|.
\end{equation}
From~$\|\matr A\matr B\|\leq\|\matr A\|_2\|\matr B\|$ it readily follows that
\[
  \|\matr Q - \matr I_d \| 
  = \|\matr\Sigma_2^{-1/4} ( \matr\Sigma_1^{1/2}-\matr\Sigma_2^{1/2} )  \matr\Sigma_2^{-1/4} \|
  \lleq \|\matr \Sigma_1^{1/2}-\matr \Sigma_2^{1/2}\|
  \lleq \|\matr \Sigma_1-\matr \Sigma_2\|^{1/2},
\]
where the last inequality is due to Theorem~1.7.2 in~\cite{aleksandrov2016operator}.

\par
The second summand on the right hand side of~\eqref{eq:HFCZwischenschritt}
is bounded from above by~$\|\matrt Q - \matr Q^{-1}\|+\|\matr Q^{-1}-\matr I_d\|$.
With the same argument as for the first summand in~\eqref{eq:HFCZwischenschritt},
we get~$\|\matr Q^{-1}-\matr I_d\| \lleq \|\matr\Sigma_1-\matr\Sigma_2\|^{1/2}$.
Since
\begin{align*}
  \|\matrt Q - \matr Q^{-1}\| 
  &= \| \matr\Sigma_2^{-1/4}( \matr\Sigma_1^{-1/4}\matr\Sigma_2\matr\Sigma_1^{-1/4}-\matr\Sigma_2^{1/2}\matr\Sigma_1^{-1/2}\matr\Sigma_2^{1/2} )\matr\Sigma_2^{-1/4} \| \\
  &\lleq \| \matr\Sigma_1^{-1/4}\matr\Sigma_2\matr\Sigma_1^{-1/4}-\matr\Sigma_2^{1/2} \|
  + \| \matr\Sigma_2^{1/2}\matr\Sigma_1^{-1/2}\matr\Sigma_2^{1/2} - \matr\Sigma_2^{1/2}\|
\end{align*}
and
\begin{align*}
  \| \matr\Sigma_1^{-1/4}\matr\Sigma_2\matr\Sigma_1^{-1/4}-\matr\Sigma_2^{1/2} \|
  &\leq \| \matr\Sigma_1^{-1/4}\matr\Sigma_2\matr\Sigma_1^{-1/4}-\matr\Sigma_1^{1/2}\|
  + \| \matr\Sigma_1^{1/2}-\matr\Sigma_2^{1/2} \|\\
  &\lleq \|\matr\Sigma_1-\matr\Sigma_2\|^{1/2}
\end{align*}
as well as
\[
  \| \matr\Sigma_2^{1/2}\matr\Sigma_1^{-1/2}\matr\Sigma_2^{1/2} - \matr\Sigma_2^{1/2}\|
  \lleq \| \matr\Sigma_1^{1/2}-\matr\Sigma_2^{1/2} \|
  \lleq\| \matr\Sigma_1-\matr\Sigma_2 \|^{1/2},
\]
we get~$\|\matrt Q - \matr I_d \| \lleq \| \matr\Sigma_1 - \matr\Sigma_2 \|^{1/2}$, concluding the proof.
\end{proof}
\end{lemma}
Now we can give an upper bound for the~$\varepsilon$-covering number in the Hellinger topology.
\begin{lemma}\label{lemma:coveringNumberContr}
Let the assumptions of Theorem~\ref{th:contraction} be fulfilled.
Then the~$\varepsilon$ covering number of~$\Theta_n$ from~\eqref{eq:sieveContraction} fulfills
\begin{equation}\label{eq:coveringContrToShow}
  \log\sup_{\varepsilon>\varepsilon_n}N(\xi\varepsilon,\{ \matr f\in\Theta_n \colon d_{n,H}(\matr f,\matr f_0)<2\varepsilon  \},d_{n,H}) \leq n\varepsilon_n^2
\end{equation}
for every~$\xi>0$.
\begin{proof}
With~$p_{j,N}$ from~\eqref{eq:pjn} it holds
\[
  d_{n,H}^2(\matr f,\matr f_0) 
  \leq \max_{j=1,\ldots,N}d_H^2( p_{j,N}(\cdot|\matr f), p_{j,N}(\cdot|\matr f_0)  )
  \lleq \|\matr f-\matr f_0\|_{F,\infty}
\]
by Lemma~\ref{lemma:hellingerFrobeniusConnection}, hence the left hand side
of~\eqref{eq:coveringContrToShow} is bounded from above by
\[
  \log N(\xi\varepsilon_n,\Theta_n,d_{n,H}) 
  \leq \log N(c\xi^4\varepsilon_n^4,\Theta_n,\|\cdot\|_{F,\infty})
  \lleq k_n(\log k_n-\log\varepsilon_n),
\]
where Lemma~\ref{lemma:coveringNumber} was used in the last step.
The rest of the proof follows along the lines of Section~9.5.2 in~\cite{ghosal2017fundamentals}.
\end{proof}
\end{lemma}

\subsubsection*{Further results}
\begin{lemma}\label{lemma:lowerHellingerBound}
Under the assumptions of Lemma~\ref{lemma:hellingerFrobeniusConnection}, it holds~$d_H^2(p_1,p_2)\ggeq\|\matr\Sigma_1-\matr\Sigma_2\|^2$
with proportionality constants depending only on~$\tau_0,\tau_1$ and~$d$.
\begin{proof}
With~$\matr A:=\matr\Sigma_2^{-1/2}\matr\Sigma_1^{1/2}$ and~$\matr B:=\matr A^*\matr A\in\mathcal S_d^+$, we compute
\[
  \frac{|2\matr\Sigma_1^{1/2}\matr\Sigma_2^{1/2}|}{|\matr\Sigma_1+\matr\Sigma_2|}
  = \frac{2^d|\matr I_d|}{|\matr A^*+\matr A^{-1}|}
  = \frac{2^d|\matr B|^{1/2}}{|\matr B + \matr I_d|}
  = \prod_{j=1}^d \frac{2\sqrt{b_j}}{1+b_j},
\]
with~$b_1,\ldots,b_d$ denoting the eigenvalues of~$\matr B$.
Since~$\matr B=\matr\Sigma_2^{-1/2}\matr\Sigma_1\matr\Sigma_2^{-1/2}$,
an application of the Courant-Fisher Min-Max principle reveals
$b_j \geq \tau_0/\tau_1>0$ and~$b_j\leq\tau_1/\tau_0$
for~$j=1,\ldots,d$ by assumption on~$\matr\Sigma_1,\matr\Sigma_2$.
From~$\prod_{j=1}^d2\sqrt b_j / (1+b_j) \leq 2 \sqrt b_l/(1+b_l)$
for~$l=1,\ldots,d$, we arrive at
\[
  d_H^2(p_1,p_2)
  = 1 - \frac{|2\matr\Sigma_1^{1/2}\matr\Sigma_2^{1/2}|}{|\matr\Sigma_1+\matr\Sigma_2|}
  \geq 1 - \frac{2\sqrt b_l}{1+b_l}
  = \frac{(b_l-1)^2}{(1+b_l)(1+\sqrt b_l)^2},
\]
which is bounded from above by~$\frac{1}{\rho} (b_l-1)^2$, 
where~$\rho=\max_{x \in [0, \tau_1/\tau_0]}((1+x)(1+\sqrt x)^2)<\infty$.
In particular, we obtain
\[
  d_H^2(p_1,p_2)
  \ggeq \max_{l=1,\ldots,d} (b_l-1)^2
  = \|\matr B - \matr I_d\|_2^2
  \ggeq \|\matr B - \matr I_d\|^2
\]
and the claim now follows from~$\|\matr B-\matr I_d\| \ggeq \| \matr \Sigma_1-\matr\Sigma_2 \|$
(see Lemma B.4~(d) in~\cite{meier}).
\end{proof}
\end{lemma}

The previous insight can be used to bound the root average squared Hellinger distance
from~\eqref{eq:hellingerFDef} for spectral densities with uniformly bounded 
Lipschitz constants from below.
\begin{lemma}\label{lemma:lowerRootHellingerBound}
Let~$\matr f_0,\matr f \in \mathcal C_{1/\tau,\tau}\cap\tilde{\mathcal C}_\tau$ for some~$\tau>0$ 
with the truncation sets from~\eqref{eq:truncationSet} and~\eqref{eq:lipschitzTruncationSet}.
Then it holds
\begin{equation}\label{eq:lrhbClaim1}
  d_{n,H}^2(\matr f_0,\matr f)
  \ggeq \| \matr f - \matr f_0 \|_{F,1}^2 + O\left( \frac{1}{n} \right)
\end{equation}
and
\begin{equation}\label{eq:lrhbClaim2}
  d_{n,H}^2(\matr f_0,\matr f)
  \ggeq \| \matr f - \matr f_0 \|_{F,\infty}^3 + O\left( \frac{1}{n} \right).
\end{equation}
\begin{proof}
Consider the function~$t(\omega):=\| \matr f(\omega) - \matr f_0(\omega) \|^2$.
From Lemma~\ref{lemma:lowerHellingerBound}, we obtain
\begin{align*}
  d_{n,H}^2(\matr f_0,\matr f)
  &= \frac{1}{N} \sum_{j=1}^N d_H^2(p_{j,N}(\cdot|\matr f_0),p_{j,N}(\cdot|\matr f))
  \ggeq \frac{1}{N} \sum_{j=1}^N t(\omega_j)\\
  &= \frac{1}{\pi}\int_0^\pi t(\omega)d\omega + O\left( \frac{1}{n} \right)
  \ggeq \|\matr f - \matr f_0 \|_{F,2}^2 + O\left( \frac{1}{n} \right)
\end{align*}
by assumption.
The bound~\eqref{eq:lrhbClaim1} now follows from H\"older's inequality.
Choose~$\omega_0$ such that $\|\matr f(\omega_0)\|=\max_{\omega}\|\matr f(\omega)\|$
and let~$A:= \{ \omega \colon |\omega-\omega_0|\leq\|\matr f\|_{F,\infty}/(2\tau) \}$.
Then it holds
\[
  \| \matr f \|_{F,2}^2 
  \geq \int_A \|\matr f(\omega)\|^2d\omega
  \geq \frac{\|\matr f\|_{F,\infty}^2}{4}\int_Ad\omega
  \geq \frac{\|\matr f\|_{F,\infty}^3}{8\tau}
  \ggeq \|\matr f\|_{F,\infty}^3,
\]
where we used~$\| \matr f(\omega) \| \geq \|\matr f\|_{F,\infty}/2$
for all~$\omega\in A$.
This concludes the proof of~\eqref{eq:lrhbClaim2}.
\end{proof}
\end{lemma}

\section{MCMC Algorithm for Posterior Computation}\label{sec:app:mcmc}
To draw random samples from the posterior distribution of~$\matr f$,
we employ a Gibbs sampler (see Section~6.3.2 in~\cite{christensen2011bayesian}) 
based on the approximate parametrization from~\eqref{eq:thetaParametrization} 
for~$\matr f$.
We assume that the Hpd Gamma process prior on~$\matr\Phi$ fulfills
the assumptions of Theorem~\ref{th:phiSeries}.
Additionally, we assume that the measure~$\alpha$ on~$\mathcal X\times\bar{\mathbb S}_d^+$
has a Lebesgue density, which will be denoted by~$g$.
The corresponding probability density will be denotes by~$g^*=g/C_\alpha$.
From Theorem~\ref{th:phiSeries}, we obtain that the prior for~$\vect\Theta_{\matr f}$
from~\eqref{eq:thetaParametrization} is given by
\[
  p(\vect\Theta_{\matr f})=p(k) |\matr J_T(\vect r)|  \exp\left( -\sum_{l=1}^L [v_l-\log g^*(x_l,\matr U_l) \right),
\]
where~$\matr J_T$ denotes the Jacobian of the transformation~$\vect r := (r_1,\ldots,r_L)\mapsto (v_1,\ldots,v_L) =: \vect v$
as in Theorem~\ref{th:phiSeries}.
It can be shown that 
\[
  \log|\matr J_T(\vect r)| = L\log C_\alpha - \sum_{l=1}^L \left( \log r_l + \beta(x_l,\matr U_l)r_l \right),
\]
see Section~5.2.1 in~\cite{meier} for the details.
The matrices~$\matr U_1,\ldots,\matr U_L \in {\mathbb S}_d^+$ are
parametrized in terms of \emph{hyperspherical coordinates} as~$\matr U_l = \matrt T \vect \varphi_l$
for~$l=1,\ldots,L$,
where~$\matrt T$ is a one-to-one transformation
and~$\vect\varphi_l$ is a~$(d^2-1)$-dimensional vector
with components in~$(0,\pi/2)$ and~$(0,\pi)$.
The Jacobian determinant of~$\matrt T$ is available analytically, 
see~(27)-(31) in~\cite{mittelbach2012sampling} or Section~3.4.1 in~\cite{meier}.

\par
Since the full conditionals are not available in a closed form, we employ
Metropolis-Hastings (MH) steps (see Section~6.3.3 in~\cite{christensen2011bayesian}) 
for the components of~$\vect\Theta_{\matr f}$, where the
starting values and proposal distributions will be discussed in the following.

For the sake of computational speed-up, the values of~$k$ can be restricted to
a finite set~$\{1,\ldots,k_{\max}\}$, where~$k_{\max}$ is some large integer.
This has the computational advantage that the Bernstein polynomial basis functions
can be precomputed and stored for usage in the iterations of the sampling algorithm.
A feasible choice of~$k_{\max}$ depends on the data at hand and can be determined
by preliminary pilot runs.
For all examples considered in this work, a value of~$k_{\max}=500$ was large enough.
To draw samples from the full conditional of~$k$, we choose some large integer as starting value.
Although samples from the full conditional could be drawn from the finite set~$\{1,\ldots,k_{\max}\}$,
we instead use a MH step to avoid the computationally expensive task
of~$k_{\max}$ conditional posterior evaluations.
The proposal is chosen as a random walk scheme with discretized Cauchy distribution
increments.

The radial parts~$r_1,\ldots,r_L$ are updated one at a time by individual MH steps
with normal random walk proposals on the log scale.
To elaborate, given the value~$r_l^{(i)}$ of~$r_l$ in iteration~$i$
of the Markov Chain, a proposal for~$r_l^{(i+1)}$ is drawn from
the~$LN( \log(r_l^{(i)}), \sigma_l^2 )$ distribution,
where the proposal variance~$\sigma_l^2$ is determined \emph{adaptively}
during burn-in (see Section~3 in~\cite{roberts2009examples}), aiming for a
target acceptance rate of~0.44.
The starting values are drawn iid from the~$\Exp(1)$ distribution.

The~$x_l$'s in~$(0,\pi)$ are also updated by individual MH steps,
with random walk proposals, where the increments are~$\operatorname{Unif}([-\delta_l,\delta_l])$
distributed with~$\delta_l = \pi l / (l+2\sqrt n)$.
The values of~$\delta_l$ are as suggested in~\cite{choudhuri} for the location
parameters of a Dirichlet process.
We also employ the author's strategy of circular proposals whenever a
proposal lands outside the interval~$[0,\pi]$.
The starting values are drawn iid from the~$\operatorname{Unif}([0,\pi])$
distribution.

The matrices~$\matr U_l$ are sampled one at a time
with a MH step for~$\vect \varphi_l$.
The components of~$\vect\varphi_l=(\varphi_{l,1},\ldots,\varphi_{l,i^2-1})$ are sampled blockwise, where the 
proposals are drawn independently 
from a random walk scheme with~$\operatorname{Unif}([-a_j\tilde\delta_l,a_j\tilde\delta_l])$ increments,
where~$a_j$ denotes the length of the interval of definition of~$\varphi_{l,j}$
and~$\tilde\delta_l>0$ is a scaling parameter that is determined adaptively
durin burn-in (again with target acceptance rate of~0.44).


\bibliographystyle{ba}
\bibliography{citations}

\begin{thebibliography}{57}
\newcommand{\enquote}[1]{``#1''}
\expandafter\ifx\csname natexlab\endcsname\relax\def\natexlab#1{#1}\fi
\expandafter\ifx\csname url\endcsname\relax
  \def\url#1{{\tt #1}}\fi
\expandafter\ifx\csname urlprefix\endcsname\relax\def\urlprefix{URL }\fi
\ifx\endbibitem\undefined \let\endbibitem\relax\fi

\bibitem[{Abramowitz and Stegun(1964)}]{abramowitz1964handbook}
Abramowitz, M. and Stegun, I.~A. (1964).
\newblock {\em Handbook of mathematical functions: with formulas, graphs, and
  mathematical tables\/}, volume~55.
\newblock Courier Corporation.
\endbibitem

\bibitem[{Akaike(1974)}]{akaike1974new}
Akaike, H. (1974).
\newblock \enquote{A new look at the statistical model identification.}
\newblock {\em IEEE transactions on automatic control\/}, 19(6): 716--723.
\endbibitem

\bibitem[{Aleksandrov and Peller(2016)}]{aleksandrov2016operator}
Aleksandrov, A.~B. and Peller, V.~V. (2016).
\newblock \enquote{Operator {L}ipschitz functions.}
\newblock {\em Russian Mathematical Surveys\/}, 71(4): 605.
\endbibitem

\bibitem[{B\"ottcher(2006)}]{boettcher}
B\"ottcher, A. (2006).
\newblock {\em Analysis of Toeplitz Operators\/}.
\newblock Springer Monographs in Mathematics. Springer-Verlag Berlin
  Heidelberg, second edition.
\endbibitem

\bibitem[{Bradley(2002)}]{bradley2002positive}
Bradley, R.~C. (2002).
\newblock \enquote{On positive spectral density functions.}
\newblock {\em Bernoulli\/}, 8(2): 175--193.
\endbibitem

\bibitem[{Brockwell and Davis(1991)}]{brockwell}
Brockwell, P. and Davis, R. (1991).
\newblock {\em Time Series: Theory and Methods\/}.
\newblock Springer Series in Statistics. Springer New York.
\endbibitem

\bibitem[{{Bureau of Meteorology of the Australian
  Government}(2018{\natexlab{a}})}]{soiGlossary}
{Bureau of Meteorology of the Australian Government} (2018{\natexlab{a}}).
\newblock \enquote{Climate Glossary -- {S}outhern {O}scillation {I}ndex
  ({SOI}).}
\newblock \url{http://www.bom.gov.au/climate/glossary/soi.shtml}.
\newblock Accessed: 2018-05-15.
\endbibitem

\bibitem[{{Bureau of Meteorology of the Australian
  Government}(2018{\natexlab{b}})}]{soiGlossary2}
--- (2018{\natexlab{b}}).
\newblock \enquote{The three phases of the {E}l {N}i\~{n}o {S}outhern
  {O}scillation ({ENSO}).}
\newblock
  \url{http://www.bom.gov.au/climate/enso/history/ln-2010-12/three-phases-of-ENSO.shtml}.
\newblock Accessed: 2018-05-15.
\endbibitem

\bibitem[{Cadonna et~al.(2017)Cadonna, Kottas, and Prado}]{Cadonna2017}
Cadonna, A., Kottas, A., and Prado, R. (2017).
\newblock \enquote{Bayesian mixture modeling for spectral density estimation.}
\newblock {\em Statistics \& Probability Letters\/}, 125: 189--195.
\endbibitem

\bibitem[{Carter and Kohn(1997)}]{CarterKohn97}
Carter, C. and Kohn, R. (1997).
\newblock \enquote{Semiparametric {B}ayesian inference for time series with
  mixed spectra.}
\newblock {\em Journal of the Royal Statistical Society: Series B (Statistical
  Methodology)\/}, 59(1): 255--268.
\endbibitem

\bibitem[{Chopin et~al.(2013)Chopin, Rousseau, and Liseo}]{Chopin13}
Chopin, N., Rousseau, J., and Liseo, B. (2013).
\newblock \enquote{Computational aspects of {B}ayesian spectral density
  estimation.}
\newblock {\em Journal of Computational and Graphical Statistics\/}, 22(3):
  533--557.
\endbibitem

\bibitem[{Choudhuri et~al.(2004{\natexlab{a}})Choudhuri, Ghosal, and
  Roy}]{choudhuri}
Choudhuri, N., Ghosal, S., and Roy, A. (2004{\natexlab{a}}).
\newblock \enquote{Bayesian estimation of the spectral density of a time
  series.}
\newblock {\em Journal of the American Statistical Association\/}, 99(468):
  1050--1059.
\endbibitem

\bibitem[{Choudhuri et~al.(2004{\natexlab{b}})Choudhuri, Ghosal, and
  Roy}]{choudhuriContiguity}
--- (2004{\natexlab{b}}).
\newblock \enquote{Contiguity of the {W}hittle measure for a {G}aussian time
  series.}
\newblock {\em Biometrika\/}, 91(1): 211--218.
\endbibitem

\bibitem[{Christensen et~al.(2011)Christensen, Johnson, Branscum, and
  Hanson}]{christensen2011bayesian}
Christensen, R., Johnson, W., Branscum, A., and Hanson, T.~E. (2011).
\newblock {\em Bayesian ideas and data analysis: an introduction for scientists
  and statisticians\/}.
\newblock CRC Press.
\endbibitem

\bibitem[{Corless et~al.(1996)Corless, Gonnet, Hare, Jeffrey, and
  Knuth}]{corless1996lambertw}
Corless, R.~M., Gonnet, G.~H., Hare, D.~E., Jeffrey, D.~J., and Knuth, D.~E.
  (1996).
\newblock \enquote{On the {L}ambert {W} function.}
\newblock {\em Advances in Computational mathematics\/}, 5(1): 329--359.
\endbibitem

\bibitem[{Dai and Guo(2004)}]{dai2004multivariate}
Dai, M. and Guo, W. (2004).
\newblock \enquote{Multivariate spectral analysis using Cholesky
  decomposition.}
\newblock {\em Biometrika\/}, 91(3): 629--643.
\endbibitem

\bibitem[{Dzhaparidze and Kotz(2012)}]{dz}
Dzhaparidze, K. and Kotz, S. (2012).
\newblock {\em Parameter Estimation and Hypothesis Testing in Spectral Analysis
  of Stationary Time Series\/}.
\newblock Springer Series in Statistics. Springer New York.
\endbibitem

\bibitem[{Edwards et~al.(2018)Edwards, Meyer, and Christensen}]{Edwards2018}
Edwards, M.~C., Meyer, R., and Christensen, N. (2018).
\newblock \enquote{Bayesian nonparametric spectral density estimation using
  {B}-spline priors.}
\newblock {\em Statistics and Computing\/}.
\newblock DOI: 10.1007/s11222-017-9796-9.
\endbibitem

\bibitem[{Gangopadhyay et~al.(1999)Gangopadhyay, Mallick, and
  Denison}]{Gango98}
Gangopadhyay, A., Mallick, B., and Denison, D. (1999).
\newblock \enquote{Estimation of spectral density of a stationary time series
  via an asymptotic representation of the periodogram.}
\newblock {\em Journal of statistical planning and inference\/}, 75(2):
  281--290.
\endbibitem

\bibitem[{Ghosal and Van Der~Vaart(2007)}]{ghosal2007convergence}
Ghosal, S. and Van Der~Vaart, A. (2007).
\newblock \enquote{Convergence rates of posterior distributions for noniid
  observations.}
\newblock {\em The Annals of Statistics\/}, 35(1): 192--223.
\endbibitem

\bibitem[{Ghosal and van~der Vaart(2017)}]{ghosal2017fundamentals}
Ghosal, S. and van~der Vaart, A. (2017).
\newblock {\em Fundamentals of nonparametric Bayesian inference\/}, volume~44.
\newblock Cambridge University Press.
\endbibitem

\bibitem[{Gray(2006)}]{gray2006toeplitz}
Gray, R.~M. (2006).
\newblock \enquote{Toeplitz and circulant matrices: A review.}
\newblock {\em Foundations and Trends in Communications and Information
  Theory\/}, 2(3): 155--239.
\endbibitem

\bibitem[{Gr{\"o}chenig(2013)}]{grochenig2013foundations}
Gr{\"o}chenig, K. (2013).
\newblock {\em Foundations of time-frequency analysis\/}.
\newblock Springer Science \& Business Media.
\endbibitem

\bibitem[{Gupta and Nagar(1999)}]{gupta}
Gupta, A. and Nagar, D. (1999).
\newblock {\em Matrix Variate Distributions\/}.
\newblock Monographs and Surveys in Pure and Applied Mathematics. Taylor \&
  Francis.
\endbibitem

\bibitem[{Hannan(1970)}]{hannan}
Hannan, E. (1970).
\newblock {\em Multiple Time Series\/}.
\newblock A Wiley publication in applied statistics. Wiley.
\endbibitem

\bibitem[{Hannan and Wahlberg(1989)}]{hannanWahlberg}
Hannan, E. and Wahlberg, B. (1989).
\newblock \enquote{Convergence rates for inverse {T}oeplitz matrix forms.}
\newblock {\em Journal of Multivariate Analysis\/}, 31(1): 127--135.
\endbibitem

\bibitem[{Hermansen(2008)}]{Hermansen08}
Hermansen, G.~H. (2008).
\newblock \enquote{Bayesian nonparametric modelling of covariance functions,
  with application to time series and spatial statistics.}
\newblock Ph.D. thesis, Universitetet i Oslo.
\endbibitem

\bibitem[{Ibragimov(1963)}]{ibragimov1963estimation}
Ibragimov, I.~A. (1963).
\newblock \enquote{On estimation of the spectral function of a stationary
  {G}aussian process.}
\newblock {\em Theory of Probability \& Its Applications\/}, 8(4): 366--401.
\endbibitem

\bibitem[{Kingman(1992)}]{kingman1992poisson}
Kingman, J. (1992).
\newblock {\em Poisson Processes\/}.
\newblock Oxford Studies in Probability. Clarendon Press.
\endbibitem

\bibitem[{Kirch et~al.(2018)Kirch, Edwards, Meier, and Meyer}]{kirchMeyer}
Kirch, C., Edwards, M.~C., Meier, A., and Meyer, R. (2018).
\newblock \enquote{Beyond Whittle: Nonparametric Correction of a Parametric
  Likelihood with a Focus on Bayesian Time Series Analysis.}
\newblock {\em Bayesian Anal.\/}.
\newblock Advance publication.
\newline\urlprefix\url{https://doi.org/10.1214/18-BA1126}
\endbibitem

\bibitem[{Koop and Korobilis(2010)}]{koop2010bayesian}
Koop, G. and Korobilis, D. (2010).
\newblock \enquote{Bayesian multivariate time series methods for empirical
  macroeconomics.}
\newblock {\em Foundations and Trends in Econometrics\/}, 3(4): 267--358.
\endbibitem

\bibitem[{L{\'e}vy(1948)}]{levy1948arithmetical}
L{\'e}vy, P. (1948).
\newblock \enquote{The arithmetical character of the Wishart distribution.}
\newblock {\em Mathematical Proceedings of the Cambridge Philosophical
  Society\/}, 44(2): 295--297.
\endbibitem

\bibitem[{Li and Krafty(2018)}]{li2017adaptive}
Li, Z. and Krafty, R.~T. (2018).
\newblock \enquote{Adaptive {B}ayesian Time-Frequency Analysis of Multivariate
  Time Series.}
\newblock {\em Journal of the American Statistical Association\/}.
\endbibitem

\bibitem[{Liseo et~al.(2001)Liseo, Marinucci, and Petrella}]{Liseo01}
Liseo, B., Marinucci, D., and Petrella, L. (2001).
\newblock \enquote{Bayesian semiparametric inference on long-range dependence.}
\newblock {\em Biometrika\/}, 88(4): 1089--1104.
\endbibitem

\bibitem[{Lorentz(2012)}]{lorentz2012bernstein}
Lorentz, G.~G. (2012).
\newblock {\em Bernstein polynomials\/}.
\newblock AMS Chelsea Publishing.
\endbibitem

\bibitem[{Marshall et~al.(2011)Marshall, Olkin, and
  Arnold}]{marshall1979inequalities}
Marshall, A.~W., Olkin, I., and Arnold, B. (2011).
\newblock {\em Inequalities: Theory of Majorization and Its Applications\/}.
\newblock Springer Series in Statistics. Springer-Verlag New York, second
  edition.
\endbibitem

\bibitem[{Mathai and Provost(2005)}]{mathai2005some}
Mathai, A. and Provost, S.~B. (2005).
\newblock \enquote{Some complex matrix-variate statistical distributions on
  rectangular matrices.}
\newblock {\em Linear Algebra and Its Applications\/}, 410: 198--216.
\endbibitem

\bibitem[{Meier(2018)}]{meier}
Meier, A. (2018).
\newblock \enquote{A Matrix Gamma Process and Applications to Bayesian Analysis
  of Multivariate Time Series.}
\newblock Ph.D. thesis, Otto-von-Guericke University Magdeburg.
\newline\urlprefix\url{https://opendata.uni-halle.de//handle/1981185920/13470}
\endbibitem

\bibitem[{Meier et~al.(2018)Meier, Kirch, Edwards, and
  Meyer}]{beyondWhittlePackage}
Meier, A., Kirch, C., Edwards, M.~C., and Meyer, R. (2018).
\newblock {\em beyondWhittle: Bayesian Spectral Inference for Stationary Time
  Series\/}.
\newblock R package version 1.1.
\endbibitem

\bibitem[{Mittelbach et~al.(2012)Mittelbach, Matthiesen, and
  Jorswieck}]{mittelbach2012sampling}
Mittelbach, M., Matthiesen, B., and Jorswieck, E.~A. (2012).
\newblock \enquote{Sampling uniformly from the set of positive definite
  matrices with trace constraint.}
\newblock {\em IEEE Transactions on Signal Processing\/}, 60(5): 2167--2179.
\endbibitem

\bibitem[{P{\'e}rez-Abreu and Rosi{\'n}ski(2007)}]{perez2007representation}
P{\'e}rez-Abreu, V. and Rosi{\'n}ski, J. (2007).
\newblock \enquote{Representation of infinitely divisible distributions on
  cones.}
\newblock {\em Journal of Theoretical Probability\/}, 20(3): 535--544.
\endbibitem

\bibitem[{P{\'e}rez-Abreu and Stelzer(2014)}]{perez}
P{\'e}rez-Abreu, V. and Stelzer, R. (2014).
\newblock \enquote{Infinitely divisible multivariate and matrix Gamma
  distributions.}
\newblock {\em Journal of Multivariate Analysis\/}, 130: 155--175.
\endbibitem

\bibitem[{Petrone(1999)}]{petrone1999random}
Petrone, S. (1999).
\newblock \enquote{Random bernstein polynomials.}
\newblock {\em Scandinavian Journal of Statistics\/}, 26(3): 373--393.
\endbibitem

\bibitem[{Roberts and Rosenthal(2009)}]{roberts2009examples}
Roberts, G.~O. and Rosenthal, J.~S. (2009).
\newblock \enquote{Examples of adaptive {MCMC}.}
\newblock {\em Journal of Computational and Graphical Statistics\/}, 18(2):
  349--367.
\endbibitem

\bibitem[{Rosen and Stoffer(2007)}]{rosen2007automatic}
Rosen, O. and Stoffer, D.~S. (2007).
\newblock \enquote{Automatic estimation of multivariate spectra via smoothing
  splines.}
\newblock {\em Biometrika\/}, 94(2).
\endbibitem

\bibitem[{Rosi{\'n}ski(2001)}]{rosinski2001series}
Rosi{\'n}ski, J. (2001).
\newblock \enquote{Series representations of {L}{\'e}vy processes from the
  perspective of point processes.}
\newblock In {\em L{\'e}vy processes\/}, 401--415. Springer.
\endbibitem

\bibitem[{Roychowdhury and Kulis(2015)}]{roychowdhuryetal2015}
Roychowdhury, A. and Kulis, B. (2015).
\newblock \enquote{Gamma Processes, Stick-Breaking and Variational Inference.}
\newblock In {\em Proceedings of The 18th International Conference on
  Artificial Intelligence and Statistics (AISTATS)\/}, 800--808.
\endbibitem

\bibitem[{Sato(1999)}]{sato1999levy}
Sato, K.-i. (1999).
\newblock {\em L{\'e}vy processes and infinitely divisible distributions\/}.
\newblock Cambridge university press.
\endbibitem

\bibitem[{Sethuraman(1994)}]{sethuraman1994constructive}
Sethuraman, J. (1994).
\newblock \enquote{A constructive definition of {D}irichlet priors.}
\newblock {\em Statistica sinica\/}, 639--650.
\endbibitem

\bibitem[{Shao and Wu(2007)}]{Shao07}
Shao, X. and Wu, B.~W. (2007).
\newblock \enquote{Asymptotic spectral theory for nonlinear time series.}
\newblock {\em Annals of Statistics\/}, 35(4): 1773--1801.
\endbibitem

\bibitem[{Shumway and Stoffer(2010)}]{shumway2010time}
Shumway, R.~H. and Stoffer, D.~S. (2010).
\newblock {\em Time series analysis and its applications: with R examples\/}.
\newblock Springer Science \& Business Media.
\endbibitem

\bibitem[{Stoffer(2017)}]{astsaR}
Stoffer, D. (2017).
\newblock {\em astsa: Applied Statistical Time Series Analysis\/}.
\newblock R package version 1.8.
\endbibitem

\bibitem[{Szab{\'o} et~al.(2015)Szab{\'o}, van~der Vaart, and van
  Zanten}]{szabo2015frequentist}
Szab{\'o}, B., van~der Vaart, A., and van Zanten, J. (2015).
\newblock \enquote{Frequentist coverage of adaptive nonparametric {B}ayesian
  credible sets.}
\newblock {\em The Annals of Statistics\/}, 43(4): 1391--1428.
\endbibitem

\bibitem[{Whittle(1957)}]{Whittle57}
Whittle, P. (1957).
\newblock \enquote{Curve and periodogram smoothing.}
\newblock {\em Journal of the Royal Statistical Society. Series B
  (Methodological)\/}, 38--63.
\endbibitem

\bibitem[{Wolpert and Ickstadt(1998)}]{wolpert1998simulation}
Wolpert, R.~L. and Ickstadt, K. (1998).
\newblock \enquote{Simulation of {L}{\'e}vy random fields.}
\newblock In {\em Practical nonparametric and semiparametric Bayesian
  statistics\/}, 227--242. Springer.
\endbibitem

\bibitem[{Zhang(2016)}]{zhang2016adaptive}
Zhang, S. (2016).
\newblock \enquote{Adaptive spectral estimation for nonstationary multivariate
  time series.}
\newblock {\em Computational Statistics \& Data Analysis\/}, 103: 330--349.
\endbibitem

\bibitem[{Zhang(2018)}]{Zhang2018}
--- (2018).
\newblock \enquote{Bayesian copula spectral analysis for stationary time
  series.}
\newblock {\em Computational Statistics and Data Analysis\/}.
\endbibitem

\end{thebibliography}

\end{document}